\documentclass[11pt]{article}
\usepackage[margin=1in]{geometry}
\usepackage{setspace}
\usepackage{amsmath,amssymb,amsthm,mathtools}
\usepackage{bm}
\usepackage{graphicx}
\graphicspath{{figures/}{./}}
\usepackage{booktabs}
\usepackage{float}
\usepackage{longtable}
\usepackage{wrapfig}
\usepackage[ruled]{algorithm}
\usepackage{algpseudocode}

\usepackage{caption}
\usepackage[round,authoryear]{natbib}
\usepackage{hyperref}
\makeatletter
\g@addto@macro\normalsize{%
  \setlength\abovedisplayskip{8pt plus 2pt minus 4pt}%
  \setlength\belowdisplayskip{8pt plus 2pt minus 4pt}%
  \setlength\abovedisplayshortskip{4pt plus 2pt}%
  \setlength\belowdisplayshortskip{6pt plus 2pt minus 3pt}%
}
\makeatother

\usepackage{enumitem}
\newtheorem{theorem}{Theorem}
\newtheorem{lemma}[theorem]{Lemma}
\newtheorem{corollary}[theorem]{Corollary}
\newtheorem{proposition}[theorem]{Proposition}

\newtheorem{assumption}[theorem]{Assumption}
\newtheorem{remark}[theorem]{Remark}
\newtheorem{definition}[theorem]{Definition}

\makeatletter\@addtoreset{theorem}{section}\makeatother

\newcommand{\E}{\mathbb{E}}
\renewcommand{\Pr}{\mathbb{P}}
\newcommand{\R}{\mathbb{R}}

\newcommand{\HH}{\mathcal{H}}

\newcommand{\norm}[1]{\lVert #1 \rVert}
\newcommand{\Opt}{\operatorname{Opt}}
\newcommand{\Optw}{\widetilde{\operatorname{Opt}}}
\newcommand{\ind}{\mathbf{1}}
\newcommand{\Var}{\operatorname{Var}}
\newcommand{\Perp}{\mathrel{\perp\!\!\!\perp}}
\title{Upper Confidence Bounds for the Prediction Error of Kernel Ridge Regression via Gaussian Refitting}
\author{Yijin Ni \\ Georgia Institute of Technology \\ \texttt{yni64@gatech.edu} \and Xiaoming Huo \\ Georgia Institute of Technology \\ \texttt{huo@gatech.edu}}
\date{}

\newcommand{\appref}[1]{Appendix~\ref{#1}}
\title{Upper Confidence Bounds for the Prediction Error of Kernel Ridge Regression via Gaussian Refitting}
\author{Yijin Ni \\ Georgia Institute of Technology \\ \texttt{yni64@gatech.edu} \and Xiaoming Huo \\ Georgia Institute of Technology \\ \texttt{huo@gatech.edu}}
\date{}

\begin{document}
\maketitle

\begin{abstract}
Assessing a single model fit requires a computable upper confidence bound for the gap between the fit and the unknown truth, as mean estimates ignore realization variance.
Standard cross-validation margins are bottlenecked at order $n^{-1/2}$ by noise fluctuations, even when the true error shrinks faster. 
While wild refitting cancels this noise level, existing Rademacher sign methods degenerate for kernel ridge regression and rely on unobservable quantities.

We propose a Gaussian refit for kernel ridge regression.
By Anderson's inequality, the fit movement is monotone in the noise sizes, yielding a computable tail bound.
Assuming only symmetric noise, the bound requires no moment assumptions and is calibrated at any confidence level via order statistics. Theoretically, using a worst-case envelope, the bound contracts at the minimax rate $O_P(n^{-2s/(2s+1)})$, correctly matching the prediction error.
Empirically, using a practical data-driven envelope, the bound maintains full coverage within twice the true $95\%$ error quantile. By contrast, cross-validation exceeds this quantile by factors up to $51$, and by hundreds under infinite-variance noise. The procedure extends empirically to nonlinear constrained estimators and real spatial data.

\end{abstract}

\medskip
\noindent\textit{Keywords:} cross-validation; excess risk; kernel ridge regression; minimax rate; upper confidence bound; wild bootstrap

\section{Introduction}
\label{sec:intro}

Kernel ridge regression, smoothing splines, and Gaussian process regression are standard methods for nonparametric prediction. Practitioners typically evaluate these models using cross-validation. This approach estimates the expected prediction error through the average loss on withheld observations.

However, estimating the expected error is insufficient when the realized error deviates heavily from the mean. Consider infinite-variance noise distributions, such as a Cauchy law. Under these conditions, the expected risk is undefined. The realized error at high quantiles vastly exceeds the median. Mean estimates fail to quantify the actual incurred error. A robust alternative is to construct high-probability upper bounds on the realized error.

In this paper, we establish these bounds in a fixed-design setting. 
This treats the covariates as fixed, restricting randomness solely to the response noise. 
Formally, given observations $\{(x_i,y_i)\}_{i=1}^n$ with fixed covariates $x_i\in\mathcal{X}$ and responses $y_i\in\mathbb{R}$ connected through an unknown true model, that is,
\begin{equation}
y_i \;=\; f^*(x_i) + w_i, \qquad i=1,\dots,n,
\label{eq:intro-model}
\end{equation}
where $f^*$ resides in a reproducing kernel Hilbert space (RKHS) $\HH$ and the noise $w_1,\dots,w_n$ is symmetric about zero.
A learning procedure $\mathcal{M}$ returns an estimator $\widehat{f}=\mathcal{M}(y)$ from the response vector $y=(y_1,\dots,y_n)$, where the fixed covariates ($x_i$'s) are directly absorbed into $\mathcal{M}$.
Denote the prediction error evaluating the average gap between $\widehat{f}$ and $f^*$ as $\mathcal{E}(\widehat f)$, that is,
\begin{equation}
\mathcal{E}(\widehat f) \;:=\; \norm{\widehat f-f^*}_n^2 \;:=\; \frac1n\sum_{i=1}^n \bigl(\widehat f(x_i)-f^*(x_i)\bigr)^2.
\label{eq:intro-target}
\end{equation}
The objective is to construct a computable upper confidence bound $\widehat{U}$, such that, 
\[
    \Pr\bigl(\mathcal{E}(\widehat f)\le\widehat U\bigr)\ge 1-\alpha,
\]
for a chosen confidence level $1-\alpha$. 
Construction of this bound requires only a bounded number of evaluations of $\mathcal{M}$, together with inputs computed from the smoother's explicit form.

Standard cross-validation cannot, in general, deliver such a bound.
It estimates prediction performance using losses evaluated on held-out observations, each of which contains fresh test noise.
Consequently, the fluctuations of the resulting estimate are governed by the noise level rather than by the prediction error of the fitted function, imposing an $n^{-1/2}$ floor on the interval margin.
By contrast, under sufficiently fast spectral decay and the corresponding regularity conditions, the prediction error of kernel ridge regression can converge faster than $n^{-1/2}$.
Proposition~\ref{prop:cv} shows that, in this regime, the margin-to-error ratio of standard cross-validation intervals diverges polynomially with $n$.
Proposition~\ref{prop:floor} establishes a more fundamental limitation: no procedure based solely on held-out losses can separate the fit's error from the test-noise fluctuation below the $n^{-1/2}$ scale.
A second limitation arises under heavy-tailed noise: the usual standard-error construction requires a finite fourth noise moment, a condition violated even by a Student-$t_4$ distribution.

Under heavy-tailed noise of unknown scale, existing methods do not provide a bound on the realized prediction error that is both computable and rate-sharp.
Point estimators target the expected error rather than its upper tail: Stein's unbiased risk estimate \citep[SURE;][]{stein1981estimation}, for example, covers the realized error with probability well below the nominal level ($0.39$--$0.58$, Table~\ref{tab:sim}).
Methods with genuine tail guarantees require assumptions absent from our setting, such as exactly Gaussian noise in the confidence-ball literature or a known noise scale in the kernel-certificate literature (\S\ref{sec:related}).
The closest alternative is the wild refit of \citet{wainwright2025wild}: flip the signs of the residuals, refit once, and estimate the error by pairing the perturbation with the displacement of the fit.
For kernel ridge regression, however, its guarantee is neither computable nor rate-sharp.
The bound depends on the noise supremum, the population bias, and an $f^*$-dependent radius, none of which is observable.
Its calibration also relies on bounded-constraint geometry absent from penalized kernel ridge regression, leaving the bound at the noise scale rather than the scale of the error.
Corollary~\ref{cor:w25krr} shows that the resulting bound-to-error ratio diverges polynomially in $n$.

Building on the refitting mechanism of \citet{wainwright2025wild}, we introduce a Gaussian refit that makes the resulting bound computable and directly calibratable.
Let $\widetilde w_i:=y_i-\widehat f(x_i)$ denote the residuals, let $|w|:=(|w_1|,\ldots,|w_n|)$ denote the unobserved noise magnitudes, and let $\circ$ denote coordinatewise multiplication.
Our construction modifies the wild refit in three respects.
First, it replaces $|w|$ by a computable vector $a$: the theoretical construction uses an envelope satisfying $a_i\ge|w_i|$, obtained from the kernel ball \eqref{eq:env-wc}, while the practical construction uses a leverage-corrected version of the residuals \eqref{eq:env-dd}.
Second, it replaces Rademacher signs by independent Gaussian multipliers
\[
\xi_1,\ldots,\xi_L\stackrel{\mathrm{iid}}{\sim}\mathcal N(0,I_n).
\]
Third, it replaces analytic tail bounds by simulation-based calibration.
For each multiplier, we refit the perturbed responses and compute
\[
m_k:=\norm{\mathcal M(y+\xi_k\circ a)-\mathcal M(y)}_n .
\]
Whereas the wild refit combines the displacement with the perturbation to estimate a cross term, we use the displacement norm directly.
For a linear smoother, the refit movement is obtained by applying the same linear operator to $\xi_k\circ a$ that maps the noise $w$ to the stochastic component of the prediction error.
Together with the envelope $a_i\ge|w_i|$, this permits the quantiles of the refit movements to calibrate an upper bound for that component.
The resulting bound is
\begin{equation}
\widehat U_\alpha
\;:=\;
\bigl(m_{(\lceil(L+1)(1-\alpha)\rceil)}+b\bigr)^2,
\label{eq:intro-bound}
\end{equation}
where $m_{(1)}\le\cdots\le m_{(L)}$ are the ordered refit movements and $b\ge0$ bounds the smoothing bias.
The same collection of refits therefore yields bounds across confidence levels without further model fitting.

The Gaussian replacement is essential: the same construction fails with Rademacher multipliers.
The envelope substitutes computable magnitudes for the unobserved $|w|$, and this substitution is valid only if enlarging a coordinatewise magnitude cannot decrease the upper quantiles of the movement.
Rademacher multipliers do not have this property: cancellation between coordinates can cause a larger magnitude to reduce the movement on half of the sign realizations (Remark~\ref{rem:twopoint}).
Gaussian multipliers do.
By Anderson's inequality, enlarging any magnitude cannot decrease any upper quantile of the movement distribution (Lemma~\ref{lem:anderson}).
A second comparison closes the chain to the truth: at the true magnitudes, the upper tail of the Gaussian movement dominates that of the sign-driven noise term (Proposition~\ref{prop:gauss-dom}).
Gaussian multipliers are therefore required by the construction, rather than chosen merely for convenience.

The Gaussian refit provides both a finite-sample guarantee and a practical error bound.
For the theoretical guarantee, we use the worst-case envelope: the resulting bound is valid conditionally on the noise magnitudes, requires no moment assumption on the noise, and contracts at the minimax rate over the kernel ball (Theorems~\ref{thm:validity} and~\ref{thm:rate}).
The analysis applies to any symmetric linear smoother, including smoothing splines and the Gaussian-process posterior mean (Remark~\ref{rem:linear}).
In practice, we replace the worst-case envelope by a sharper vector built from leverage-corrected residuals and evaluate the resulting bound empirically.
Across the sample sizes and noise laws considered, it achieves full coverage while remaining within twice the true $95\%$ error quantile; the three cross-validation intervals are $4$ to $51$ times that quantile (Table~\ref{tab:sim} and Figure~\ref{fig:sim}).
Under Cauchy noise, the Gaussian-refit bound remains in single digits, whereas the cross-validation bounds reach the hundreds (\S\ref{sec:sim-moment}).
We further apply the practical procedure to a nonlinear constrained fit and a real elevation field, where it achieves coverage of $0.99$--$1.00$ in both cases (Tables~\ref{tab:port} and~\ref{tab:realdata}).

\paragraph{Contributions.}
The paper makes four contributions:
\begin{enumerate}[label=(\roman*), leftmargin=*, itemsep=2pt, topsep=2pt]
\item \textbf{Limits of held-out losses.}
We prove that cross-validation intervals are polynomially wider than the realized prediction error whenever the fit converges faster than $n^{-1/2}$, and that no bound computed from held-out losses can cross the $n^{-1/2}$ floor, however the interval is repaired (Propositions~\ref{prop:cv} and~\ref{prop:floor}).
\item \textbf{Necessity of Gaussian multipliers.}
We prove that the Rademacher wild refit is not rate-sharp for kernel ridge regression, and we isolate the two properties a multiplier must supply for the construction to close: monotonicity of the movement quantiles in the envelope, and domination of the noise term's upper tail. Gaussian multipliers supply both; signs supply neither (Corollary~\ref{cor:w25krr}, Remark~\ref{rem:twopoint}, and \S\ref{sec:completion}).
\item \textbf{A computable, rate-sharp bound.}
We construct the bound \eqref{eq:intro-bound} from a bounded number of refits, and we prove that it is valid in finite samples conditionally on the noise magnitudes, under no moment assumption on the noise, and that it contracts at the minimax rate over the kernel ball. We also give the sharper leverage-corrected envelope used in practice (\S\ref{sec:the-bound}, \S\ref{sec:data-driven}, and Theorems~\ref{thm:validity}--\ref{thm:rate}).
\item \textbf{Empirical evaluation.}
We evaluate the bound against three forms of cross-validation, nested cross-validation, SURE, and the Rademacher statistic, across sample sizes and noise laws down to Cauchy, and we carry it to a nonlinear constrained fit and a real elevation field (\S\ref{sec:simulation}--\S\ref{sec:realdata}).
\end{enumerate}

The scope of the guarantees is as follows. The proofs cover linear smoothers with the penalty fixed in advance (Remarks~\ref{rem:linear} and~\ref{rem:lambda}); the extension to nonlinear constrained fitting is empirical. Validity is conservative and relies on the noise being sufficiently delocalized by the smoother. The worst-case envelope is loose for slowly decaying kernel spectra, and the guarantees assume that the target belongs to the reproducing kernel Hilbert space, failing which cross-validation is the appropriate tool.

\section{Setup}
\label{sec:setup}

\begin{table}[t]
\centering
\caption{Notation used throughout the paper. Symbols local to \S\ref{sec:wr-inequality} are defined where they appear.}
\label{tab:notation}
{\setstretch{1}\footnotesize
\begin{tabular}{@{}l@{\quad}p{0.68\textwidth}@{}}
\toprule
\multicolumn{2}{@{}l}{\textit{Data and model}}\\
\midrule
$n$ & sample size \\
$x_i$ & covariate, $x_i\in\mathcal{X}\subseteq\R^d$ \\
$y_i$ & response, $y_i\in\R$ \\
$f^*$ & the unknown regression function, $y_i=f^*(x_i)+w_i$, $\forall i$ \\
$w,\,|w|,\,\varepsilon$ & noise $w=(w_1,\dots,w_n)$, its absolute value, and its signs; $w_i=\varepsilon_i|w_i|$, $\forall i$ \\
$\HH,\,k$ & RKHS containing $f^*$, induced by reproducing kernel $k$ \\
$\norm{\cdot}_\HH,\,\langle\cdot,\cdot\rangle_\HH$ & norm and inner product induced by the RKHS $\HH$ \\
$B,\,\kappa$ & known bounds for $f^*$ and $k$, i.e.\ $\norm{f^*}_\HH\le B$ and $\sup_x k(x,x)\le\kappa^2$ \\
\midrule
\multicolumn{2}{@{}l}{\textit{Empirical norm and pairing}}\\
\midrule
$\norm{\cdot}_n$ & empirical norm, $\norm{g}_n^2=\tfrac1n\sum_i g(x_i)^2$ \\
$\langle\cdot,\cdot\rangle_n$ & empirical pairing, $\langle g,v\rangle_n=\tfrac1n\sum_i g(x_i)v_i$ \\
\midrule
\multicolumn{2}{@{}l}{\textit{Estimator, target, and goal}}\\
\midrule
$\mathcal{M}$ & the estimator, $\mathcal{M}:\R^n\to\HH$; a function $g$ enters through $(g(x_i))_{i=1}^n$ \\
$\widehat f$ & the fit of $f^*$, $\widehat f=\mathcal{M}(y)$ \\
$\mathcal{E}(\widehat f)$ & excess risk, $\norm{\widehat f-f^*}_n^2$ \\
$\widehat U_\alpha$ & upper confidence bound at level $1-\alpha$, i.e.\ $\Pr(\mathcal{E}(\widehat f)\le\widehat U_\alpha)\ge 1-\alpha$ \\
\midrule
\multicolumn{2}{@{}l}{\textit{Kernel ridge regression, the running instance}}\\
\midrule
$K$ & Gram matrix, $K=(k(x_i,x_j))_{ij}$ \\
$\lambda$ & ridge penalty \\
$H$ & kernel ridge smoother, $H=K(K+n\lambda I)^{-1}$; $\widehat f=Hy$ \\
$\mu_j$ & eigenvalues of $K/n$, $\mu_1\ge\cdots\ge0$ \\
$\asymp$ & equality up to constant factors \\
$s$ & smoothness index, $\mu_j\asymp j^{-2s}$ \\
$d_n$ & effective dimension, $d_n=\operatorname{tr}(H)$ \\
$\check f,\,R$ & constrained fit \eqref{eq:constrained} and its radius, $\norm{f}_\HH\le R$ \\
\midrule
\multicolumn{2}{@{}l}{\textit{The Gaussian-refit bound} (\S\ref{sec:certificate}), $\widehat U_\alpha=(q_{1-\alpha}(a)+b)^2$}\\
\midrule
$\xi$ & Gaussian multiplier, $\xi\sim\mathcal N(0,I_n)$ \\
$\circ$ & coordinatewise product, $(u\circ v)_i=u_i v_i$ \\
$a$ & envelope replacing the unknown $|w|$; $a_i\ge|w_i|$, $\forall i$ \\
$m(a)$ & movement of the fit under one draw, $\norm{\mathcal{M}(y+\xi\circ a)-\mathcal{M}(y)}_n$ \\
$L$ & number of refit draws \\
$q_{1-\alpha}(a)$ & noise term, the $\lceil(L+1)(1-\alpha)\rceil$-th smallest of the $L$ movements \\
$b$ & bias input, any $b\ge\norm{(H-I)f^*}_n$ \\
\bottomrule
\end{tabular}}
\end{table}

Given samples $\{(x_i,y_i)\}_{i=1}^n$ with fixed input covariates $x_i\in\mathcal{X}\subseteq\R^d$ and scalar responses $y_i\in\R$ for each $i$, suppose there exists an unknown regression function $f^*$ in a reproducing kernel Hilbert space (RKHS) $\HH$, with kernel $k$ and norm $\norm{\cdot}_\HH$, building the relationship between the covariates and the responses, that is,
\begin{equation}
y_i \;=\; f^*(x_i) + w_i, \qquad i=1,\dots,n,
\label{eq:model}
\end{equation}
where $w=(w_1,\dots,w_n)$ is the observation noise. Let $\widehat f$ be the fitted estimator of $f^*$, a transformation of the observed responses $y=(y_1,\dots,y_n)$ produced by the training procedure, that is, $\widehat f=\mathcal{M}(y)$, where $\mathcal{M}:\R^n\to\HH$ denotes the transformation. We identify a function $g\in\HH$ with its value-vector $(g(x_i))_{i=1}^n$ at the design points; then $\mathcal{M}(g)$ denotes $\mathcal{M}$ applied to those values, and functions and vectors combine coordinatewise there, so that $\mathcal{M}$ acts on a fit or on $f^*$ as well as on the response $y$. Consider the instance-wise excess risk as the prediction error of the fitted model $\widehat f$, that is,
\begin{equation}
\mathcal{E}(\widehat f) \;:=\; \norm{\widehat f - f^*}_n^2 \;=\; \tfrac1n\textstyle\sum_i \bigl(\widehat f(x_i)-f^*(x_i)\bigr)^2 ,
\label{eq:target}
\end{equation}
where $\norm{\cdot}_n$ is the empirical norm, $\norm{g}_n^2:=\tfrac1n\sum_i g(x_i)^2$ for any $g:\mathcal{X}\to\R$. In this work, we seek a computable upper confidence bound $\widehat U_\alpha$ for the excess risk $\mathcal{E}(\widehat f)$ at a user-chosen level $1-\alpha$, that is,
\begin{equation}
\Pr\bigl(\mathcal{E}(\widehat f)\le\widehat U_\alpha\bigr)\ge 1-\alpha, \qquad \alpha\in(0,1) .
\label{eq:goal}
\end{equation}

Two assumptions constrain the model and the noise: the bounded target and kernel of Assumption~\ref{as:model}, which feed the worst-case envelope and bias input of \S\ref{sec:the-bound}, and the conditional symmetry of Assumption~\ref{as:noise}, the sole distributional requirement, which licenses the sign perturbation. Table~\ref{tab:notation} collects the notation. We write $\langle g,v\rangle_n:=\tfrac1n\sum_i g(x_i)v_i$ for the pairing of a function $g$ with a vector $v\in\R^n$ at the design points.

\begin{assumption}[Bounded target and kernel]
\label{as:model}
The regression function has bounded norm and the kernel is bounded:
\[
\norm{f^*}_\HH\le B, \qquad \sup_{x\in\mathcal{X}} k(x,x)\le\kappa^2 ,
\]
for known constants $B$ and $\kappa$.
\end{assumption}

\begin{assumption}[Conditional symmetry of the noise]
\label{as:noise}
Decompose the noise vector into its magnitude and sign, i.e., $w_i=\varepsilon_i|w_i|$ with $\varepsilon_i:=\operatorname{sign}(w_i)\in\{-1,1\}$, and $|w|:=(|w_1|,\dots,|w_n|)$. Suppose the signs are i.i.d.\ and independent of the magnitudes, i.e.,
\[
\varepsilon\Perp|w|, \qquad \Pr(\varepsilon_i=+1)=\Pr(\varepsilon_i=-1)=\tfrac12, \quad \forall i .
\]
\end{assumption}

Equivalently, the sign vector is uniform on the hypercube and independent of the magnitudes, $\varepsilon\sim\operatorname{Unif}\{-1,1\}^n$. The magnitudes $|w|$ are otherwise arbitrary: they may be dependent, heavy-tailed, or of varying scale across the observations. The noise need not have any finite moment, and $\E[w_i^2]$ may be infinite. This is condition~(11b) of \citet{wainwright2025wild}.

\paragraph{The estimators.}
We study two estimators, both firmly non-expansive. The first, our object of study, is \emph{kernel ridge regression},
\begin{equation}
\widehat f \;:=\; \arg\min_{f\in\HH}\Bigl\{\tfrac1n\textstyle\sum_i(y_i-f(x_i))^2+\lambda\norm{f}_\HH^2\Bigr\} \;=\; Hy,
\qquad H \;:=\; K(K+n\lambda I)^{-1},
\label{eq:krr}
\end{equation}
the penalized least-squares linear smoother with Gram matrix $K=(k(x_i,x_j))_{ij}$; we write $\widehat f$ interchangeably for the fitted function $\mathcal{M}(y)\in\HH$ and its vector of values $Hy=(\widehat f(x_i))_{i=1}^n$ at the design points. The smoother is symmetric with eigenvalues in $[0,1]$; write $\mu_1\ge\cdots\ge\mu_n\ge0$ for the eigenvalues of $K/n$ and $d_n:=\operatorname{tr}(H)=\sum_j\mu_j/(\mu_j+\lambda)$ for the effective dimension. Every guarantee in this paper is for \eqref{eq:krr}.
The second, its nonlinear counterpart, is the reproducing-kernel-ball--constrained least-squares fit
\begin{equation}
\check f \;:=\; \arg\min_{\norm{f}_\HH\le R}\ \tfrac1n\textstyle\sum_i\bigl(y_i-f(x_i)\bigr)^2 ,
\label{eq:constrained}
\end{equation}
the projection of the data onto the ball $\{f:\norm{f}_\HH\le R\}$; this is the setting of the wild-refit theory of \citet{wainwright2025wild}, and we use it in \S\ref{sec:sim-main} to test our construction empirically. We take $R\ge B$, so that $f^*$ is feasible for \eqref{eq:constrained}.
Both \eqref{eq:krr} and \eqref{eq:constrained} are firmly non-expansive in the empirical norm, $\norm{\mathcal{M}(y)-\mathcal{M}(y')}_n^2\le\langle\mathcal{M}(y)-\mathcal{M}(y'),\,y-y'\rangle_n$: kernel ridge through the linear form $H$, since $H-H^2=H(I-H)$ is positive semidefinite, and the constrained fit through the projection. This property underlies the movement query of \S\ref{sec:certificate} for the linear smoother and the portability of the construction to the constrained fit in \S\ref{sec:data-driven}. The guarantees of \S\ref{sec:theory} use the linear structure of \eqref{eq:krr}, and \eqref{eq:constrained} marks the boundary of what the construction reaches beyond it.

\section{The calibrated wild-refit bound}
\label{sec:certificate}

We build on the wild-refit inequality of \citet{wainwright2025wild}, which bounds the prediction error by a computable movement plus a remainder; for the kernel ridge smoother that remainder is uncomputable and the bound degenerates (\S\ref{sec:wr-inequality}). We therefore replace the optimism decomposition by a direct noise-and-bias split, and bound its noise term with Gaussian rather than Rademacher multipliers: Anderson's inequality licenses a computable envelope, and calibrating the movement by an order statistic yields $\widehat U_\alpha=(q_{1-\alpha}(a)+b)^2$ (\S\ref{sec:completion}). We then give it in a worst-case version, proved valid and rate-optimal in \S\ref{sec:theory} (\S\ref{sec:the-bound}), and a tighter data-driven version evaluated in \S\ref{sec:simulation}, whose reach beyond the linear smoother we test on the constrained fit \eqref{eq:constrained} (\S\ref{sec:data-driven}).

\subsection{The wild-refit inequality of Wainwright}
\label{sec:wr-inequality}

Recall the kernel ridge fit \eqref{eq:krr}. We seek a high-probability bound on the excess risk $\mathcal{E}(\widehat f)$; the wild refit of \citet{wainwright2025wild} supplies one, and specializing it to kernel ridge shows why a new construction is needed. Decompose the error into an estimation part and an approximation part,
\begin{equation}
\mathcal{E}(\widehat f)\;\le\;2\norm{\widehat f-f^\dagger}_n^2+2\norm{f^\dagger-f^*}_n^2 ,
\qquad
f^\dagger \;:=\; \mathcal{M}(f^*) \;=\; Hf^* ,
\label{eq:err-split}
\end{equation}
where $f^\dagger$, the \emph{best penalized approximation}, is the fit the estimator returns from the noiseless response, so that $\norm{f^\dagger-f^*}_n=\norm{(H-I)f^*}_n$ is the approximation bias. The estimation part is driven by the noise through a cross term: firm non-expansiveness gives $\norm{\widehat f-f^\dagger}_n^2\le\langle\widehat f-f^\dagger,\,w\rangle_n$, unobservable since it involves $w$ and $f^*$. Bounding such cross terms is the role of the wild refit: it perturbs the residuals $\widetilde w:=y-\widehat f$ with i.i.d.\ Rademacher signs $\varepsilon$ and refits at a scale $\rho>0$, giving the \emph{wild refit} and the computable \emph{wild optimism}
\begin{equation}
\widehat f^{\bullet}_\rho \;:=\; \mathcal{M}\bigl(\widehat f+\rho\,\varepsilon\circ\widetilde w\bigr),
\qquad
\Optw^{\bullet}_\rho \;:=\; \tfrac1n\textstyle\sum_i \varepsilon_i\widetilde w_i\bigl(\widehat f^{\bullet}_\rho(x_i)-\widehat f(x_i)\bigr).
\label{eq:wild-refit}
\end{equation}
Two results of \citet{wainwright2025wild}, restated in \appref{app:degeneracy}, turn the wild refit into a bound. The first bounds the \emph{optimism} $\Opt^\star(\widehat f):=\tfrac1n\sum_i w_i(\widehat f(x_i)-f^*(x_i))$, the cross term at the fit, by the wild optimism plus a remainder. The second bounds the estimation error by the \emph{critical radius} $r^*$, the fixed point $W_n(r^*)\asymp(r^*)^2$ of the \emph{wild complexity} $W_n(r):=\sup_{\norm{f-\widehat f}_n\le r}\tfrac1n\sum_i\varepsilon_i\widetilde w_i(f(x_i)-\widehat f(x_i))$, the cross term localized to an $r$-ball, whose growth sets the rate. Wainwright's theory is developed for the constrained fit \eqref{eq:constrained} and defers the penalized smoother \eqref{eq:krr}; specializing both results there shows that neither ingredient survives.

\begin{corollary}[Degeneracy of the wild-refit bound]
\label{cor:w25krr}
Take the linear smoother \eqref{eq:krr} with the pilot equal to the fit, so that the residuals are $\widetilde w=(I-H)y$ and the best penalized approximation is $f^\dagger=Hf^*$.

\emph{(i) The optimism bound is uncomputable.} For any $t>0$, at the scale $\rho$ matched to a radius $r\ge\norm{\widehat f-f^\dagger}_n$ by $\norm{\widehat f^{\bullet}_\rho-\widehat f}_n=2r$, with probability at least $1-4\exp(-t^2)$,
\begin{equation}
\Opt^\star(\widehat f)\;\le\;
\underbrace{\Optw^{\bullet}_\rho}_{\text{computable}}
\;+\;
\underbrace{\bigl(3r+\norm{(H-I)f^*}_n\bigr)\tfrac{2\norm{w}_\infty t}{\sqrt n}+A_n}_{\text{not computable}} ,
\label{eq:w25krr}
\end{equation}
whose remainder involves the radius $r$ (unobservable through $f^\dagger$), the noise sup-norm $\norm{w}_\infty$, the bias $\norm{(H-I)f^*}_n$, and a pilot-approximation supremum $A_n$, none computable from the data.

\emph{(ii) The rate degenerates.} The wild complexity is a linear maximization, $W_n(r)=\sup_{\norm{g}_n\le r}\langle\varepsilon\circ\widetilde w,\,g\rangle_n=r\norm{\widetilde w}_n$, so the critical radius is the residual norm, $r^*\asymp\norm{(I-H)y}_n$, of constant order, and the certified bound reads, with probability at least $1-4\exp(-t^2)$,
\begin{equation}
\norm{\widehat f-f^*}_n\;\lesssim\;\norm{(I-H)y}_n ,
\label{eq:w25krr-rate}
\end{equation}
against the truth $\norm{\widehat f-f^*}_n=O_P(n^{-s/(2s+1)})$ of Lemma~\ref{lem:krr-rate}: a gap of order at least $n^{s/(2s+1)}\to\infty$. The restated propositions and the derivation are in \appref{app:degeneracy}.
\end{corollary}

The deficiency of \eqref{eq:w25krr} lies in the inputs, not the inequalities: the concentration bounds behind the wild-refit theory apply to the Rademacher chaos without difficulty, but they are phrased in the noise sup-norm, the population bias, and a radius reachable only through $f^\dagger$, with generic constants, so no sharper tail inequality repairs the bound; what is missing is an observable surrogate for those inputs, and \S\ref{sec:completion} shows that under two-point multipliers no surrogate can carry a guarantee, their law not being monotone in its magnitudes. The wild complexity keys off the geometry of a constraint set, which the linear smoother lacks; kernel ridge calls instead for a construction tailored to its linear structure, which we give next.

\subsection{Completion and calibration}
\label{sec:completion}

Corollary~\ref{cor:w25krr} closes the optimism route for the kernel ridge smoother, but its linear structure makes the detour unnecessary. Since $\widehat f=Hy$ and $f^\dagger=Hf^*$ \eqref{eq:err-split}, the estimation error of \S\ref{sec:wr-inequality} is explicit, $\widehat f-f^\dagger=Hw$, so the triangle inequality behind \eqref{eq:err-split} reads
\begin{equation}
\norm{\widehat f-f^*}_n \;\le\; \underbrace{\norm{Hw}_n}_{=\;\norm{\widehat f-f^\dagger}_n} \;+\; b ,
\label{eq:wr-ineq}
\end{equation}
with $b\ge\norm{(H-I)f^*}_n=\norm{f^\dagger-f^*}_n$ an input bounding the approximation bias of \eqref{eq:err-split}. The noise term $\norm{Hw}_n$ is the estimation error that the wild-refit theory reached only through the optimism and its critical radius, delivered directly by the linear smoother; it is still unobservable, so we bound it by simulation, perturbing with Gaussian multipliers at a computable envelope, incurring an explicit remainder in place of the uncomputable one of \eqref{eq:w25krr}.

\paragraph{Completion.}
To bound the noise term we simulate it. Since $\norm{Hw}_n=\norm{H(\varepsilon\circ|w|)}_n$ is driven by the unknown magnitudes $|w|$ and signs $\varepsilon$, we replace both: the magnitudes by a computable \emph{envelope} $a\ge|w|$, and the signs by Gaussian multipliers $\xi$. The estimator's response to the synthetic noise is read through a \emph{movement query}, a single refit at $y+\xi\circ a$, giving the \emph{movement} $m(a)=\norm{H(\xi\circ a)}_n$. The completion rests on the chain
\[
\norm{Hw}_n \;=\; \norm{H(\varepsilon\circ|w|)}_n
\;\underset{\text{signs}\,\to\,\xi}{\preceq}\; \norm{H(\xi\circ|w|)}_n
\;\underset{|w|\,\to\,a}{\preceq}\; \norm{H(\xi\circ a)}_n \;=\; m(a),
\]
where $\preceq$ denotes domination of the upper tail. The first step is strict: a sign-refit at the true magnitudes would reproduce the law of $\norm{Hw}_n$ exactly, leaving no margin, while the Gaussian's extra dispersion lifts the upper tail (Proposition~\ref{prop:gauss-dom}, up to an explicit remainder). The second step is Anderson's inequality (Lemma~\ref{lem:anderson}); neither step survives sign multipliers (Remark~\ref{rem:twopoint}). The law of $m(a)$ is exactly known for the linear smoother (\appref{app:exact-law}), so the calibrated threshold carries no analytic constants. The movement query, the envelope, and the two dominations are made precise next.

\begin{definition}[Movement query]
\label{def:movement}
For a perturbation $v\in\R^n$, the movement query returns the displacement of the
fit about the base point $\mathcal{M}(y)$,
\[
\Delta_{\mathcal{M}}(v)\;:=\;\mathcal{M}(y+v)-\mathcal{M}(y),
\]
which for the linear smoother \eqref{eq:krr} equals $Hv$.
\end{definition}

\begin{definition}[Noise envelope]
\label{def:envelope}
A noise envelope is a vector $a\in\R_{\ge0}^n$ dominating the noise magnitudes coordinatewise,
\begin{equation}
a_i\ge|w_i|\quad\text{for every }i .
\label{eq:envelope-cond}
\end{equation}
\end{definition}

Draw Gaussian multipliers $\xi\sim\mathcal N(0,I_n)$ and, for a noise envelope $a$, form the refit movement
\begin{equation}
m(a) \;:=\; \norm{\Delta_{\mathcal{M}}(\xi\circ a)}_n ,
\label{eq:refit}
\end{equation}
which for a linear smoother is $m(a)=\norm{H(\xi\circ a)}_n$.

\emph{First, the magnitudes.} At an envelope satisfying \eqref{eq:envelope-cond}, inflating them to $a\ge|w|$ can only enlarge the movement, by a property of Gaussian measures with no counterpart for two-point multipliers.

\begin{lemma}[Envelope domination]
\label{lem:anderson}
Let $u,v\in\R^n$ be fixed with $u_i\ge v_i\ge 0$ for every $i$, and $\xi\sim\mathcal N(0,I_n)$. For any matrix $A$ and every $t\ge0$,
\[
\Pr\bigl(\norm{A(\xi\circ u)}> t\bigr) \;\ge\; \Pr\bigl(\norm{A(\xi\circ v)}> t\bigr).
\]
\end{lemma}

This is Anderson's inequality \citep{anderson1955integral}, proved in \appref{app:proof-anderson}. For Rademacher multipliers it fails: with $u=(1,1)$, $v=(1,0)$ and $A=[1,1]$ the two signs cancel and $\Pr(|\varepsilon_1+\varepsilon_2|>\tfrac12)=\tfrac12<1$. The Gaussian is the multiplier under which the completion is available.

\begin{remark}[No sign multiplier is envelope-monotone]
\label{rem:twopoint}
The failure is not particular to the example or to its scale: any symmetric two-point multiplier is a scaled Rademacher law, and the two coordinates cancel identically at every scale $\rho>0$, so the counterexample persists for the whole family. The monotonicity of the movement in the magnitudes, which Lemma~\ref{lem:anderson} supplies for Gaussian multipliers, is therefore unavailable to sign-based refits altogether: no envelope substitution under a two-point multiplier carries a one-sided guarantee.
\end{remark}

\emph{Second, the signs.} With the magnitudes handled by the envelope, the remaining difference from the noise term is the signs: at the true magnitudes, the Gaussian $\xi$ replaces the real signs $\varepsilon$ and keeps the movement conservative.

\begin{proposition}[Gaussian domination of the signed-noise term]
\label{prop:gauss-dom}
Fix the magnitudes $|w|$ and draw the signs $\varepsilon\sim\operatorname{Unif}\{-1,1\}^n$
of Assumption~\ref{as:noise}. The squared noise term and the squared Gaussian
movement at the true magnitudes,
\[
\norm{Hw}_n^2=\tfrac1n(\varepsilon\circ|w|)^\top H^\top H(\varepsilon\circ|w|),
\qquad
m(|w|)^2=\tfrac1n(\xi\circ|w|)^\top H^\top H(\xi\circ|w|),
\]
are quadratic forms in $H^\top H$ with the same mean, and the Gaussian form has the
larger variance. Consequently the Gaussian movement is conservative in the upper
tail: for every $\alpha\in(0,\tfrac12)$, the $(1-\alpha)$-quantile $q_{1-\alpha}(m(|w|))$ of the movement satisfies
\[
\Pr\bigl(\norm{Hw}_n>q_{1-\alpha}(m(|w|))\bigr)\;\le\;\alpha+\Delta_n,
\]
where $\Delta_n$ is the Berry--Esseen remainder made explicit in
Theorem~\ref{thm:validity}.
\end{proposition}

The two steps then chain to the threshold. By the magnitude inflation of Lemma~\ref{lem:anderson}, the $(1-\alpha)$-quantile of $m(a)$ at any envelope satisfying \eqref{eq:envelope-cond} is at least $q_{1-\alpha}(m(|w|))$, which dominates $\norm{Hw}_n$ by Proposition~\ref{prop:gauss-dom} up to $\Delta_n$; so that quantile of $m(a)$ is a conservative threshold for the noise term.

\paragraph{Calibration.}\label{sec:calibration}
Analytic tail bounds on $m(a)$ are conservative, so the threshold is read from the simulation directly.

\begin{definition}[Calibrated threshold]
\label{def:threshold}
Draw $\xi_1,\dots,\xi_L\sim\mathcal N(0,I_n)$, let $m_k$ be the movement \eqref{eq:refit} at $\xi_k$, and set the calibrated threshold to the order statistic
\begin{equation}
q_{1-\alpha}(a) \;:=\; m_{(k^*)}, \qquad k^*:=\lceil (L+1)(1-\alpha)\rceil ,
\label{eq:orderstat}
\end{equation}
the $k^*$-th smallest of the $L$ movements, which estimates the $(1-\alpha)$-quantile of $m(a)$ at every level from one set of draws.
\end{definition}

A single set of $L$ draws calibrates every level; the Monte Carlo error contributes the $C_1/\sqrt L$ term of Theorem~\ref{thm:validity}, and the exact weighted-chi-square law of $m(a)$ for the linear smoother is recorded in \appref{app:exact-law}.

Collecting the split \eqref{eq:wr-ineq} and the calibrated threshold gives our bound.

\begin{corollary}[The calibrated wild-refit bound]
\label{cor:bound}
Suppose the envelope satisfies \eqref{eq:envelope-cond} and the bias input obeys $b\ge\norm{(H-I)f^*}_n$. On the event $\{\norm{Hw}_n\le q_{1-\alpha}(a)\}$, the split \eqref{eq:wr-ineq} gives
\[
\norm{\widehat f-f^*}_n \;\le\; \norm{Hw}_n + b \;\le\; q_{1-\alpha}(a)+b ,
\]
and squaring yields
\begin{equation}
\mathcal{E}(\widehat f) \;\le\; \widehat U_\alpha \;:=\; \bigl(q_{1-\alpha}(a)+b\bigr)^2 .
\label{eq:ucb}
\end{equation}
\end{corollary}

The bound \eqref{eq:ucb} is \eqref{eq:intro-bound} in population form: the order statistic \eqref{eq:orderstat} estimates the quantile $q_{1-\alpha}(a)$ from the $L$ draws.

By the envelope-inflation of Lemma~\ref{lem:anderson} and the Gaussian domination of Proposition~\ref{prop:gauss-dom}, the calibrated quantile $q_{1-\alpha}(a)$ is a conservative threshold for the noise term, so the event of Corollary~\ref{cor:bound} holds with high probability; Theorem~\ref{thm:validity} makes this quantitative, bounding $\Pr(\mathcal{E}(\widehat f)>\widehat U_\alpha\mid|w|)$ by $\alpha$ up to explicit delocalization and Monte Carlo remainders. The bound is computed by Algorithm~\ref{alg:cert} with $L$ refits, and it remains to choose the envelope $a$ and the matching bias input $b$, where rigour and tightness trade off; we use two.

\subsection{The theoretical bound}
\label{sec:the-bound}

The worst-case envelope meets \eqref{eq:envelope-cond} through a closed-form bound on the fit error. Writing $g_i:=\widehat f(x_i)-f^*(x_i)$, the reproducing property, Cauchy--Schwarz, and the bound $\norm{f^*}_\HH\le B$ of Assumption~\ref{as:model} give $|g_i|=\bigl|\langle\widehat f-f^*,\,k(x_i,\cdot)\rangle_\HH\bigr|\le\bigl(\norm{\widehat f}_\HH+B\bigr)\sqrt{k(x_i,x_i)}$, so with $w_i=\widetilde w_i+g_i$ the envelope and bias input
\begin{equation}
a_i:=|\widetilde w_i|+S_i,\quad S_i:=\bigl(\norm{\widehat f}_\HH+B\bigr)\sqrt{k(x_i,x_i)},\qquad b:=\tfrac12 B\sqrt\lambda,
\label{eq:env-wc}
\end{equation}
satisfy $a_i\ge|w_i|$ and $b\ge\norm{(H-I)f^*}_n$; both are derived in \appref{app:envelope}, the latter from $\norm{(H-I)f^*}_n^2\le\tfrac14\lambda B^2$. Substituting into \eqref{eq:ucb} gives the closed-form theoretical bound
\begin{equation}
\widehat U_\alpha=\Bigl(q_{1-\alpha}(a)+\tfrac12 B\sqrt\lambda\Bigr)^2,\qquad a_i=|\widetilde w_i|+\bigl(\norm{\widehat f}_\HH+B\bigr)\sqrt{k(x_i,x_i)},
\label{eq:ucb-wc}
\end{equation}
computable from the fit and the known radius $B$ of Assumption~\ref{as:model}, and carrying the guarantee of \S\ref{sec:theory}. When no a-priori radius is available, the data-driven bound of \S\ref{sec:data-driven} replaces the worst-case inputs by estimates; the two constructions instantiate the same $\widehat U_\alpha$ at two envelopes, this one certified in \S\ref{sec:theory}, that one measured in \S\ref{sec:simulation}.

\subsection{The data-driven bound}
\label{sec:data-driven}

\paragraph{The data-driven envelope.}
For use in practice we replace the worst-case bounds by estimates. The noise scale uses the leverage correction of the wild-bootstrap literature \citep{mammen1993bootstrap,davidson2008wild}, and the bias input plugs an undersmoothed pilot fit in for the unknown $f^*$,
\begin{equation}
a_i:=\frac{|\widetilde w_i|}{1-h_{ii}}+|\widehat b_i|, \qquad
b:=\norm{\widehat b}_n, \qquad \widehat b:=(H-I)\widehat f_{\mathrm{pil}},
\label{eq:env-dd}
\end{equation}
with $h_{ii}$ the $i$th diagonal of the smoother $H$ of \eqref{eq:krr} and $\widehat f_{\mathrm{pil}}:=H_{\mathrm{pil}}\,y$ a kernel ridge fit at a smaller penalty, so it tracks $f^*$ more closely and carries less bias. Substituting into \eqref{eq:ucb} gives the bound evaluated throughout the paper,
\begin{equation}
\widehat U_\alpha=\Bigl(q_{1-\alpha}(a)+\norm{\widehat b}_n\Bigr)^2,\qquad
a_i=\frac{|\widetilde w_i|}{1-h_{ii}}+|\widehat b_i| .
\label{eq:ucb-dd}
\end{equation}
These substitutions are estimates rather than upper bounds, so $a$ does not satisfy \eqref{eq:envelope-cond} deterministically; the bound is substantially tighter than \eqref{eq:ucb-wc} but only empirically valid, its coverage established in \S\ref{sec:simulation}.

The same bound applies beyond the linear smoother. For the constrained fit \eqref{eq:constrained}, a nonlinear firmly non-expansive estimator, the movement query of Definition~\ref{def:movement} is evaluated directly on $\check f$; the leverage correction and the bias input $b$, which have no closed form for a nonlinear fit, are taken from a linear kernel-ridge stand-in at the same penalty. The construction is thus computable there too; its coverage, which \S\ref{sec:theory} does not certify, is established empirically in \S\ref{sec:sim-main}.

\begin{algorithm}[t]
\caption{Calibrated Gaussian refit}
\label{alg:cert}
\begin{algorithmic}[1]
\State \textbf{input:} data, fit $\widehat f$, envelope $a$, bias input $b$, level $1-\alpha$, draws $L$
\State base fit $\widehat f \gets \mathcal{M}(y)$
\For{$k=1,\dots,L$}
  \State draw $\xi_k\sim\mathcal N(0,I_n)$; \; $m_k \gets \norm{\mathcal{M}(y+\xi_k\circ a)-\mathcal{M}(y)}_n$
\EndFor
\State \textbf{return} $\widehat U_\alpha \gets \bigl(m_{(\lceil(L+1)(1-\alpha)\rceil)}+b\bigr)^2$
\end{algorithmic}
\end{algorithm}

Algorithm~\ref{alg:cert} runs with $(a,b)$ from \eqref{eq:env-wc} for the theoretical bound of \S\ref{sec:the-bound} or from \eqref{eq:env-dd} for the data-driven bound of \S\ref{sec:data-driven}: one procedure, two inputs.

\section{Validity and rate optimality}
\label{sec:theory}

Throughout this section $\mathcal{M}$ is the kernel ridge smoother \eqref{eq:krr} and $\widehat U_\alpha$ is the bound \eqref{eq:ucb} with the worst-case envelope of \S\ref{sec:the-bound}, drawn with $L$ Gaussian vectors.
Conditional on the noise magnitudes $|w|$, set
\begin{equation}
M \;:=\; \operatorname{diag}(|w|)\,H^\top H\,\operatorname{diag}(|w|),\qquad N \;:=\; M-\operatorname{diag}(M),
\label{eq:M}
\end{equation}
so that the squared refit movement at the true magnitudes is $m(|w|)^2=\tfrac1n\xi^\top M\xi$, and define the delocalization functionals
\begin{equation}
\delta \;:=\; \frac{\lambda_{\max}(M)}{\norm{M}_F},\qquad
\rho \;:=\; \frac{\bigl(\sum_i M_{ii}^2\bigr)^{1/2}}{\norm{M}_F},\qquad
\delta_N \;:=\; \frac{\norm{N}_{\mathrm{op}}}{\norm{N}_F},
\label{eq:deloc}
\end{equation}
where $\lambda_{\max}$ is the largest eigenvalue, $\norm{\cdot}_F$ the Frobenius norm, $\norm{\cdot}_{\mathrm{op}}$ the spectral norm, and $M_{ii}$ the $i$th diagonal entry; the three are related by $\delta_N\le(\delta+\rho)/\sqrt{1-\rho^2}$.
Throughout, $z_{1-\alpha}$ denotes the standard normal quantile, $\phi$ its density, and $O_P,\Omega_P$ the usual stochastic orders.

\subsection{Validity}

\begin{theorem}[Validity]
\label{thm:validity}
Under Assumptions~\ref{as:model} and~\ref{as:noise}, there are universal constants $C_0,C_1$ such that, conditional on the noise magnitudes $|w|$, for every $\alpha\in(0,\tfrac12)$ and every $L\ge1$, with $\delta,\delta_N$ as in \eqref{eq:deloc},
\begin{equation}
\Pr\bigl(\mathcal{E}(\widehat f) > \widehat U_\alpha \,\big|\, |w|\bigr)
\;\le\;
\alpha \;+\; \frac{C_0}{\phi(z_{1-\alpha})}\,\bigl(\delta+\delta_N\bigr) \;+\; \frac{C_1}{\sqrt L}.
\label{eq:validity}
\end{equation}
\end{theorem}

The bound \eqref{eq:validity} is non-asymptotic and explicit in the delocalization functionals \eqref{eq:deloc} and the number of refits.
Whenever $\delta\to0$ and $\rho\to0$, so that $\delta_N\to0$, and $L\to\infty$, the right-hand side tends to $\alpha$; in that regime the variance deficit of the sign statistic makes the bound conservative.
The proof combines the envelope domination of Lemma~\ref{lem:anderson} with a Berry--Esseen comparison of the Gaussian and sign quadratic forms in $M$: both have the same mean, the Gaussian has the larger variance, so its quantile is a conservative threshold, and the functionals \eqref{eq:deloc} control the normal approximation of each.
When the noise concentrates on a few coordinates, so that $\delta\to1$, the approximation fails and \eqref{eq:validity} is vacuous.
The proof is in \appref{app:proof-validity}.

Theorem~\ref{thm:validity} locates the freedom from moment assumptions.
Conditional on $|w|$, the noise term $\norm{Hw}_n=\norm{H(\varepsilon\circ|w|)}_n$ is a bounded function of $n$ fair signs, with no tails to control; the marginal law of the noise enters only when \eqref{eq:validity} is integrated over the magnitudes,
\[
\Pr\bigl(\mathcal{E}(\widehat f)>\widehat U_\alpha\bigr)
\;=\;\E\,\Pr\bigl(\mathcal{E}(\widehat f)>\widehat U_\alpha\,\big|\,|w|\bigr)
\;\le\;\alpha\;+\;\frac{C_0}{\phi(z_{1-\alpha})}\,\E\bigl[\delta+\delta_N\bigr]
\;+\;\frac{C_1}{\sqrt L},
\]
valid for every distribution of the magnitudes, Cauchy included, since $\delta,\delta_N\in[0,1]$ make the expectation exist under any noise law.
The level term $\alpha$ passes through the expectation untouched; the tails enter only through $\E[\delta+\delta_N]$, the frequency with which a realization concentrates on few coordinates.
What the injected Gaussian dominates is therefore the conditional law of the signs at the given magnitudes, not the tails of the noise, which no Gaussian could dominate; the tails ride on the realized $|w|$ and enter both sides of the comparison equally, the error through $w$ and the bound through the envelope \eqref{eq:env-wc}.

\begin{remark}[Beyond kernel ridge]
\label{rem:linear}
The proof of Theorem~\ref{thm:validity} uses the smoother only through the linear representation $\widehat f=Hy$ and the quadratic form \eqref{eq:M}; it therefore holds verbatim for any symmetric linear smoother $0\preceq H\preceq I$, given inputs $(a,b)$ satisfying \eqref{eq:envelope-cond} and $b\ge\norm{(H-I)f^*}_n$. The worst-case construction of these inputs in \S\ref{sec:the-bound} uses only the reproducing property, so it applies to every penalized reproducing-kernel smoother $H=K(K+n\lambda I)^{-1}$; smoothing splines and the Gaussian-process posterior mean are of this form, with their own kernels in place of $k$ and, for the latter, the noise variance in the role of $n\lambda$. Theorem~\ref{thm:rate} then reads Assumption~\ref{as:decay} on the corresponding kernel's spectrum.
\end{remark}

\begin{remark}[Selecting the penalty]
\label{rem:lambda}
The guarantees treat the penalty $\lambda$ as fixed in advance, which is what keeps $H$ a constant matrix. Selecting $\lambda$ on an independent split leaves every statement intact, since conditional on the split the smoother is again a fixed linear map; selecting it on the same data makes $H$ a function of $y$, the fit is then no longer linear, and the guarantees do not apply. The closest case with evidence is the ball-constrained fit of \S\ref{sec:data-driven}, which by duality is kernel ridge at a data-dependent penalty: its coverage in Table~\ref{tab:port} suggests mild data dependence is tolerated, but no guarantee is claimed.
\end{remark}

\subsection{Rate optimality}

\begin{assumption}[Eigendecay and penalty]
\label{as:decay}
There are constants $0<c_1\le c_2$ and $s>\tfrac12$ with $c_1 j^{-2s}\le\mu_j\le c_2 j^{-2s}$ for $1\le j\le n$, the penalty satisfies $\lambda\asymp n^{-2s/(2s+1)}$, and the smoother has bounded leverage, $\max_i (H^\top H)_{ii}\le C_D\,d_n/n$ for a constant $C_D$.
\end{assumption}

Under Assumption~\ref{as:decay} the effective dimension satisfies $d_n\asymp n^{1/(2s+1)}$. The results of this and the next subsection, unlike Theorem~\ref{thm:validity}, involve the scale of the noise.

\begin{assumption}[Noise scale]
\label{as:scale}
There is a constant $\sigma$ such that $\Pr\bigl(|w_i|>t\bigr)\le 2\exp\{-t^2/(2\sigma^2)\}$ for every $t>0$ and every $i$.
\end{assumption}

Assumption~\ref{as:scale} enters no statement about the level: validity is free of moment conditions, and the scale governs only how fast the bound contracts, so under heavy tails the bound remains valid but loose, the regime measured in \S\ref{sec:sim-moment}. Under Assumptions~\ref{as:model}, \ref{as:noise}, \ref{as:decay}, and~\ref{as:scale} the fitted estimator attains the minimax rate.

\begin{lemma}[Kernel-ridge rate]
\label{lem:krr-rate}
Under Assumptions~\ref{as:model}, \ref{as:noise}, \ref{as:decay}, and~\ref{as:scale}, $\mathcal{E}(\widehat f)=\norm{\widehat f-f^*}_n^2=O_P\bigl(n^{-2s/(2s+1)}\bigr)$.
\end{lemma}

The proof is in \appref{app:proof-rate}. The next result shows the bound contracts at this same rate, so it tracks the prediction error rather than sitting at a fixed multiple above it; the efficiency statement \eqref{eq:efficiency} pairs it with Lemma~\ref{lem:krr-rate}.

\begin{assumption}[Noise energy]
\label{as:energy}
There is a constant $c_0>0$ with $\tfrac1n\sum_{i=1}^n |w_i|^2 (H^\top H)_{ii}\ge c_0\, d_n/n$.
\end{assumption}

\begin{theorem}[Rate optimality]
\label{thm:rate}
Under Assumptions~\ref{as:model}, \ref{as:noise}, \ref{as:decay}, and~\ref{as:scale},
\begin{equation}
\widehat U_\alpha \;=\; O_P\bigl(n^{-2s/(2s+1)}\bigr),
\label{eq:rate}
\end{equation}
the minimax rate of estimation over $\{f\in\HH:\norm{f}_\HH\le B\}$. If Assumption~\ref{as:energy} also holds, then
\begin{equation}
\widehat U_\alpha \,/\, \mathcal{E}(\widehat f) \;=\; O_P(1).
\label{eq:efficiency}
\end{equation}
\end{theorem}

Assumption~\ref{as:energy} requires the noise to carry non-negligible energy through the smoother; it holds whenever the per-coordinate noise scale is bounded below.
The minimax rate is that of \citet{stone1982optimal} and \citet{yang1999information}; the proof is in \appref{app:proof-rate}.
For kernels with faster-than-polynomial eigendecay, such as the Gaussian, $d_n$ is bounded, Assumption~\ref{as:decay} fails, and \eqref{eq:rate}--\eqref{eq:efficiency} hold with an additional factor of $\log n$.

\subsection{Comparison with cross-validation}

We compare against the hold-out cross-validation upper bound
\begin{equation}
\widehat U_\alpha^{\mathrm{cv}} \;:=\; \widehat{\mathrm{Err}} + z_{1-\alpha}\,\widehat{\mathrm{se}} - \widehat\sigma^2 ,
\label{eq:cvbound}
\end{equation}
where $\widehat{\mathrm{Err}}$ is the mean held-out squared loss, $\widehat{\mathrm{se}}$ its standard error over the held-out points, and $\widehat\sigma^2$ a noise-variance estimate that places the bound on the prediction-error scale.

\begin{assumption}[Noise spread]
\label{as:spread}
The noise has a finite fourth moment, and the empirical variance of the squared held-out noise is bounded below: $\operatorname{empvar}_i\{w_i^2\}\ge\kappa_0>0$ for a constant $\kappa_0$.
\end{assumption}

\begin{proposition}[Cross-validation's margin does not keep pace]
\label{prop:cv}
Under Assumptions~\ref{as:model}, \ref{as:noise}, \ref{as:decay}, \ref{as:scale}, and~\ref{as:spread}, the margin $z_{1-\alpha}\,\widehat{\mathrm{se}}$ of \eqref{eq:cvbound} satisfies
\begin{equation}
\begin{aligned}
z_{1-\alpha}\,\widehat{\mathrm{se}} &\;=\; \Omega_P\bigl(n^{-1/2}\bigr)
\quad\text{independently of the fit, whence}\\
\frac{z_{1-\alpha}\,\widehat{\mathrm{se}}}{\mathcal{E}(\widehat f)} &\;=\; \Omega_P\bigl(n^{(2s-1)/(2(2s+1))}\bigr)\;\longrightarrow\;\infty .
\end{aligned}
\label{eq:cvdiverge}
\end{equation}
\end{proposition}

The mechanism fits in one display. With $\delta_i$ the fit's error at a held-out point and $m\asymp n$ the hold-out size, a held-out loss decomposes as
\[
\ell_i=(w_i+\delta_i)^2=w_i^2+2w_i\delta_i+\delta_i^2,
\qquad
z_{1-\alpha}\,\widehat{\mathrm{se}}
\;\asymp\;\frac{\operatorname{sd}_i\{\ell_i\}}{\sqrt m}\;\gtrsim\;n^{-1/2},
\]
a floor set by the spread of the leading term $w_i^2$, which no quality of fit can reduce: the fit enters only through $\delta_i$, and $\delta_i\to0$ drives the margin to the floor, not past it.
The prediction error, by contrast, contracts at $n^{-2s/(2s+1)}$ (Lemma~\ref{lem:krr-rate}), so the ratio in \eqref{eq:cvdiverge} diverges at exponent $(2s-1)/(2(2s+1))$: for $s>\tfrac12$ the fit converges faster than the parametric rate at which a noise level can be learned; for $s\le\tfrac12$, and in the bias-dominated regimes of \appref{app:sim-scope}, the error remains at the noise scale and cross-validation keeps pace.
Neither dimension nor signal-to-noise ratio enters: under homoscedastic Gaussian noise the noise level cancels from the ratio, which is driven by the smoothness and the sample size alone.
The failure concerns levels, not comparisons: differences of fold losses cancel the common noise average, which is why cross-validation remains consistent for model selection \citep{wager2020cv,lei2020cvc} while its level bound diverges.
The finite fourth moment of Assumption~\ref{as:spread} makes $\widehat{\mathrm{se}}$ well defined, and the lower bound $\operatorname{empvar}_i\{w_i^2\}\ge\kappa_0$ fails only when $w_i^2$ is constant across observations, the one case in which the margin keeps pace.
The proof is in \appref{app:proof-cv}.

\begin{corollary}
\label{cor:iid}
If $w_1,\dots,w_n$ are independent and sub-Gaussian with $\Var(w_i)\ge\tau_0^2>0$ and $\Var(w_i^2)\ge\kappa_0>0$, and are independent of the design, then Assumption~\ref{as:scale} holds by definition, Assumptions~\ref{as:energy} and~\ref{as:spread} hold with probability tending to one, and the conclusions \eqref{eq:efficiency} and \eqref{eq:cvdiverge} hold unconditionally.
\end{corollary}

Proposition~\ref{prop:cv} concerns one interval, and the interval can be repaired: the standard error can be estimated by the nested construction, or the normal approximation abandoned for an order statistic of held-out errors over repeated splits.
The next result shows that no repair escapes, because the floor is informational rather than distributional: a held-out loss reveals the noise and the fit's error only through the square $(w_i+\delta_i)^2$, and a problem with slightly larger fit error and slightly smaller noise generates held-out losses statistically indistinguishable from the original, so a bound valid for both problems must sit above the larger error.

\begin{proposition}[Holdout floor]
\label{prop:floor}
Condition on the training responses, so that the held-out error profile $\delta_i:=f^*(x_i)-\widehat f^{\mathrm{tr}}(x_i)$, $i\in V$, $|V|=m\asymp n$, is fixed, write $\norm{\delta}_m^2:=m^{-1}\sum_{i\in V}\delta_i^2$, and let the held-out noise be Gaussian, $w_i\sim\mathcal N(0,\sigma^2)$.
Let $\widehat U:=g(\ell_1,\dots,\ell_m)$ be any measurable function of the held-out losses $\ell_i:=(y_i-\widehat f^{\mathrm{tr}}(x_i))^2$ that is valid over a neighbourhood of problems,
\begin{equation}
\Pr_{\delta,\sigma}\bigl(\widehat U\ge\norm{\delta}_m^2\bigr)\;\ge\;1-\alpha
\quad\text{for all }(\delta,\sigma)\text{ with }
\norm{\delta}_m^2\le\sigma_0^2,\;
\norm{\delta}_\infty\le C_\infty\sigma_0,\;
\sigma\in[\tfrac12\sigma_0,2\sigma_0].
\label{eq:floor-class}
\end{equation}
Then for every $\kappa>0$ and every sequence $\rho_n\to0$ there is a sequence $\tau_n(\kappa)\to0$ such that, for every problem in the class with $\norm{\delta}_m^2\le\rho_n$, $\norm{\delta}_\infty\le\tfrac12 C_\infty\sigma_0$ and $\sigma\in[\sigma_0,\tfrac32\sigma_0]$,
\begin{equation}
\Pr\bigl(\widehat U\ge\kappa\,\sigma_0^2\,m^{-1/2}\bigr)\;\ge\;1-\alpha-\tau_n(\kappa) .
\label{eq:floor}
\end{equation}
\end{proposition}

Under the assumptions of Lemma~\ref{lem:krr-rate} the training-fold profile satisfies $\norm{\delta}_m^2=O_P(n^{-2s/(2s+1)})$, so a sequence $\rho_n\to0$ can be chosen slowly enough that the restriction $\norm{\delta}_m^2\le\rho_n$ holds with probability tending to one, and dividing \eqref{eq:floor} by $\mathcal{E}(\widehat f)=O_P(n^{-2s/(2s+1)})$ recovers the divergence rate of \eqref{eq:cvdiverge} for the whole class at once.
The interval \eqref{eq:cvbound}, its nested refinement, and the order-statistic repair evaluated in \S\ref{sec:simulation} are all functions of held-out losses alone, so each faces the dichotomy: it either fails the validity requirement \eqref{eq:floor-class} or obeys the floor \eqref{eq:floor}.
The proof (\appref{app:proof-floor}) is a two-point argument along the moment ridge $\delta_i^2+\sigma^2=\mathrm{const}$: the law of the held-out losses is nearly invariant along the ridge, because squaring erases the sign information that would separate the two problems, while the target $\norm{\delta}_m^2$ moves by $\kappa\sigma_0^2m^{-1/2}$; the machinery is that of the quadratic-functional testing bounds of \citet{caiLow2006adaptive}.
The Gaussian refit lies outside the class \eqref{eq:floor-class}: it reads not held-out losses but the movement of the fit under injected perturbations, and along the ridge the movement separates the two problems, which is why the floor does not bind it.
Note also that \eqref{eq:floor} is stated under Gaussian noise: the floor needs no heavy tails, so the fourth-moment requirement of Assumption~\ref{as:spread} is a second, separate deficiency of the interval \eqref{eq:cvbound}, not the source of the divergence.

\section{Related work}
\label{sec:related}

Across the literatures adjacent to our question one pattern recurs: the centre of the error distribution can be estimated under weak assumptions, but every existing tail statement for the realized error is purchased either with Gaussian noise or with a known noise scale.
No prior method delivers a finite-sample, level-$\alpha$, computable upper confidence bound on the realized empirical-norm error of the same-data kernel fit under conditional symmetry alone; the three groups below account for the near misses.

\paragraph{Resampling and held-out methods.}
The wild bootstrap perturbs residuals by independent multipliers to reproduce heteroscedastic noise, with classical validity for regression \citep{wu1986jackknife,liu1988bootstrap,mammen1992when,mammen1993bootstrap,davidson2008wild}. \citet{wainwright2025wild} turn the device into an error bound for a firmly non-expansive fit, with extensions to Bregman losses, asymmetric noise, and subsampled refits \citep{huSimchiLevi2025bregman,huSimchiLevi2025doubly}; its mechanism, the degeneracy at kernel ridge, and the Gaussian completion that repairs it are the subject of \S\ref{sec:wr-inequality}--\S\ref{sec:completion} (Corollary~\ref{cor:w25krr}). For cross-validation \citep{stone1974cross,geisser1975predictive}, the across-fold standard error is biased, honest intervals need the nested construction of \citet{batesHastieTibshirani2024cv}, the noise level is hard to estimate from within the sample \citep{bengioGrandvalet2004no}, and the refinements of \citet{bayle2020cv} and \citet{austernZhou2025cv} sharpen the variance estimate but keep the mean-plus-standard-error form, whose floor is Propositions~\ref{prop:cv} and~\ref{prop:floor}. \citet{wager2020cv} observes that the leading fluctuation of cross-validation is model-independent, cancelling in comparisons but not in levels, which is why model selection survives \citep{lei2020cvc} while level bounds do not, the distinction drawn after Proposition~\ref{prop:cv}.

\paragraph{Risk estimation and confidence sets.}
Stein's unbiased risk estimate \citep{stein1981estimation} recovers the risk of a weakly differentiable estimator; \citet{li1989honest} inverted it into honest confidence balls for the realized loss, \citet{bellecZhang2021sure} quantify its fluctuation, and both are tied to Gaussian noise with known or estimable variance. A confidence ball centred at the fit is an upper confidence bound on its realized error, the classical antecedent of our question \citep{li1989honest,beranDumbgen1998modulation,juditskyLambertLacroix2003,baraud2004confidence,caiLow2006adaptive,robinsVanDerVaart2006adaptive}; every construction assumes Gaussian or moment-bounded noise with known, or interval-known, variance, and even \citet{robinsVanDerVaart2006adaptive}, with an arbitrary centring estimator, require sample splitting and a known variance, their radius carrying the $n^{-1/2}$ term of Proposition~\ref{prop:floor}. Two boundaries locate our contribution: with unknown noise level and no shape restriction, honest balls of nontrivial radius do not exist \citep{baraud2004confidence}, and conditional symmetry is the structural assumption we show suffices; honesty over a smoothness scale caps the radius at $n^{-1/4}$ \citep{li1989honest,caiLow2006adaptive}, while our bound contracts at the minimax rate (Theorem~\ref{thm:rate}), its radius pinned to a known reproducing-kernel ball. A control-oriented literature likewise certifies $|\widehat f(x)-f^*(x)|$ from a known kernel-norm bound and a known noise envelope, sub-Gaussian \citep{abbasiYadkori2011improved,chowdhuryGopalan2017kernelized,fiedler2021practical}, bounded \citep{maddalena2021deterministic}, or energy-bounded \citep{lahr2025optimal}; heavy-tailed noise admits no such envelope, and our bound is calibrated from the data, with the noise scale nowhere an input.

\paragraph{Adjacent inference targets.}
Exactness of the sign distribution for symmetric noise also powers the sign-perturbed-sums method \citep{csajiCampiWeyer2015sps} and its kernel extensions \citep{csajiKis2019kernel}, whose exact, distribution-free regions cover regression parameters or ideal noise-free representations under structural assumptions such as a known input law or finite variance, whereas we bound the realized empirical-norm error of the deployed fit, with neither. Conformal prediction \citep{vovk2005algorithmic,lei2018distribution,barber2021predictive} is finite-sample valid for a future response under exchangeability; its target is the next observation, not the error of the fitted regression function, so the two are complementary. Kernel-ridge bands via multiplier bootstrap \citep{singhVijaykumar2023kernel} are asymptotic, need bounded residuals, and target the function rather than the realized error.

The proof ingredients have their own lineage: the minimax benchmark over a reproducing-kernel ball is classical \citep{stone1982optimal,yang1999information}, with sharp kernel-ridge rates in the effective-dimension parametrization \citep{caponnettoDeVito2007optimal,linCevherRosasco2018optimal}, and Theorem~\ref{thm:rate} shows our bound attains it. The completion rests on Gaussian comparison \citep{anderson1955integral} and Berry--Esseen bounds for quadratic forms \citep{nourdinPeccatiReinert2010,doeblerPeccati2017}, and the lower-bound machinery of Proposition~\ref{prop:floor} is that of quadratic-functional testing \citep{laurentMassart2000,caiLow2006adaptive}.

\section{Simulation: accuracy and coverage of the data-driven bound}
\label{sec:simulation}

The theory of \S\ref{sec:theory} certifies the worst-case envelope; here we test the sharper data-driven envelope, whose coverage is empirical. We ask two questions, one per subsection: does the bound stay accurate and covered as the sample grows (\S\ref{sec:sim-main}); and, as a closing stress test, does it survive noise with no moments (\S\ref{sec:sim-moment}). The boundary of the method, slower eigendecay and misspecification, is charted in \appref{app:sim-scope} and summarized at the end of \S\ref{sec:sim-main}.

\paragraph{Design.}
The data follow the model \eqref{eq:model} on a fixed uniform design,
\[
x_i=\frac{i-1}{n-1},\qquad
f^*(x)=\sin(2\pi x)+\tfrac12\sin(6\pi x),\qquad
n\in\{500,\,2000,\,8000\},
\]
with radial basis kernel of bandwidth $\gamma=50$ and ridge penalty $\lambda=10^{-3}$; the bound uses $L=199$ Gaussian refits and the order statistic at $k^*=190$, with $R$ replicates per cell. \appref{app:sim-details} gives the calibration in full.

\paragraph{Comparators.}
Cross-validation is the comparator of record, the tool in common use for this question, and enters in two forms: the interval \eqref{eq:cvbound} with the noise level $\bar\sigma^2$ known, the form Proposition~\ref{prop:cv} analyses, and with the noise level estimated by the df-corrected residual variance
\[
\widehat\sigma^2:=\frac{\norm{y-\widehat f}^2}{n-2\operatorname{tr}(H)+\operatorname{tr}(H^2)},
\]
the form a practitioner computes. The remaining comparators are SURE \citep{stein1981estimation} and the Rademacher wild-refit statistic \citep{wainwright2025wild}, both point estimates. In the small-sample cells of \S\ref{sec:sim-main} we add the Bates--Hastie--Tibshirani nested cross-validation \citep{batesHastieTibshirani2024cv}, the strongest cross-validation interval available, where its finite-sample advantage over the plain hold-out bound is largest and where cross-validation is most competitive.

We summarise each method by the metrics of \eqref{eq:metrics}: accuracy, the median bound over the true $95\%$ prediction-error quantile, and coverage, the fraction of replicates on which the bound exceeds the realized error. A confidence bound at level $0.95$ has coverage at least $0.95$ and, ideally, accuracy near one.

\begin{equation}
\mathrm{accuracy} \;:=\; \frac{\operatorname{med}(\widehat U)}{q_{0.95}},
\qquad
\mathrm{coverage} \;:=\; \frac1R\sum_{r=1}^R \ind\bigl\{\widehat U^{(r)}\ge\mathcal{E}^{(r)}\bigr\},
\label{eq:metrics}
\end{equation}

\subsection{Accuracy and coverage across sample sizes}
\label{sec:sim-main}

Table~\ref{tab:sim} reports the comparison. The Gaussian refit sits at $1.7$--$1.9\times$ the true $95\%$ quantile with full coverage on every cell. Cross-validation's margin diverges in both forms, exactly as Propositions~\ref{prop:cv} and~\ref{prop:floor} predict: with the noise level known it grows from $4.3\times$ at $n=500$ to $10.5$--$21\times$ at $n=8000$, and its coverage falls from $0.96$ toward $0.85$ as the noise loses moments; with the noise level estimated the margin is unchanged, growing to $10$--$28\times$, while coverage is full, because the residual estimate tracks the realized noise energy and the subtraction self-centres. The dichotomy of Proposition~\ref{prop:floor} is visible in the data: the level of the interval is repairable, its noise-scale margin is not. SURE and the Rademacher statistic are point estimates, with accuracy below one, at times negative, and coverage below $0.58$; they are not confidence bounds.

\begin{table}[t]
\centering
\caption{Accuracy and coverage \eqref{eq:metrics} across sample sizes and noise laws, radial basis kernel, $R$ replicates; the Cauchy row reports the median of bound over realized error (\S\ref{sec:sim-moment}), since the quantile is outlier-dominated there. Cross-validation appears in its known-level, estimated-level, and quantile-repair forms (\S\ref{sec:simulation}); SURE and the Rademacher statistic are point estimates, not bounds. The tightest method with coverage at least $0.95$ is in bold.}
\label{tab:sim}
\setstretch{1}\footnotesize
\setlength{\tabcolsep}{3.5pt}
\begin{tabular}{@{}llcccccc@{}}
\toprule
noise & $n$ & Gaussian refit & CV ($\bar\sigma^2$) & CV ($\widehat\sigma^2$) & quantile CV & SURE & Rademacher \\
\midrule
Gaussian & $500$ & $\bm{1.73\ (1.00)}$ & $4.33\ (0.92)$ & $4.56\ (1.00)$ & $7.60\ (1.00)$ & $0.37\ (0.49)$ & $0.23\ (0.36)$ \\
Gaussian & $2000$ & $\bm{1.83\ (1.00)}$ & $7.61\ (0.93)$ & $7.07\ (1.00)$ & $13.5\ (1.00)$ & $0.95\ (0.58)$ & $0.31\ (0.36)$ \\
Gaussian & $8000$ & $\bm{1.83\ (1.00)}$ & $10.5\ (0.96)$ & $10.1\ (1.00)$ & $19.1\ (1.00)$ & $-0.14\ (0.44)$ & $0.42\ (0.39)$ \\
Laplace & $500$ & $\bm{1.88\ (1.00)}$ & $6.08\ (0.91)$ & $6.25\ (1.00)$ & $11.1\ (1.00)$ & $0.36\ (0.45)$ & $0.33\ (0.33)$ \\
Laplace & $2000$ & $\bm{1.89\ (1.00)}$ & $11.3\ (0.93)$ & $11.1\ (1.00)$ & $20.2\ (1.00)$ & $0.30\ (0.46)$ & $0.30\ (0.31)$ \\
Laplace & $8000$ & $\bm{1.92\ (1.00)}$ & $18.2\ (0.98)$ & $18.9\ (1.00)$ & $31.1\ (1.00)$ & $0.36\ (0.49)$ & $0.36\ (0.32)$ \\
Student-$t_4$ & $500$ & $\bm{1.88\ (1.00)}$ & $6.67\ (0.88)$ & $7.61\ (1.00)$ & $13.3\ (1.00)$ & $0.11\ (0.43)$ & $0.27\ (0.32)$ \\
Student-$t_4$ & $2000$ & $\bm{1.76\ (1.00)}$ & $13.6\ (0.90)$ & $14.8\ (1.00)$ & $26.6\ (1.00)$ & $-0.40\ (0.43)$ & $0.28\ (0.30)$ \\
Student-$t_4$ & $8000$ & $\bm{1.84\ (1.00)}$ & $21.3\ (0.85)$ & $28.4\ (1.00)$ & $50.7\ (1.00)$ & $-3.93\ (0.39)$ & $0.27\ (0.29)$ \\
\addlinespace
Cauchy & $2000$ & $\bm{4.64\ (1.00)}$ & $503\ (1.00)$ & $286\ (1.00)$ & $547\ (1.00)$ & $204\ (1.00)$ & $0.61\ (0.22)$ \\
\bottomrule
\end{tabular}
\end{table}

\begin{figure}[t]
\centering
\includegraphics[width=\linewidth]{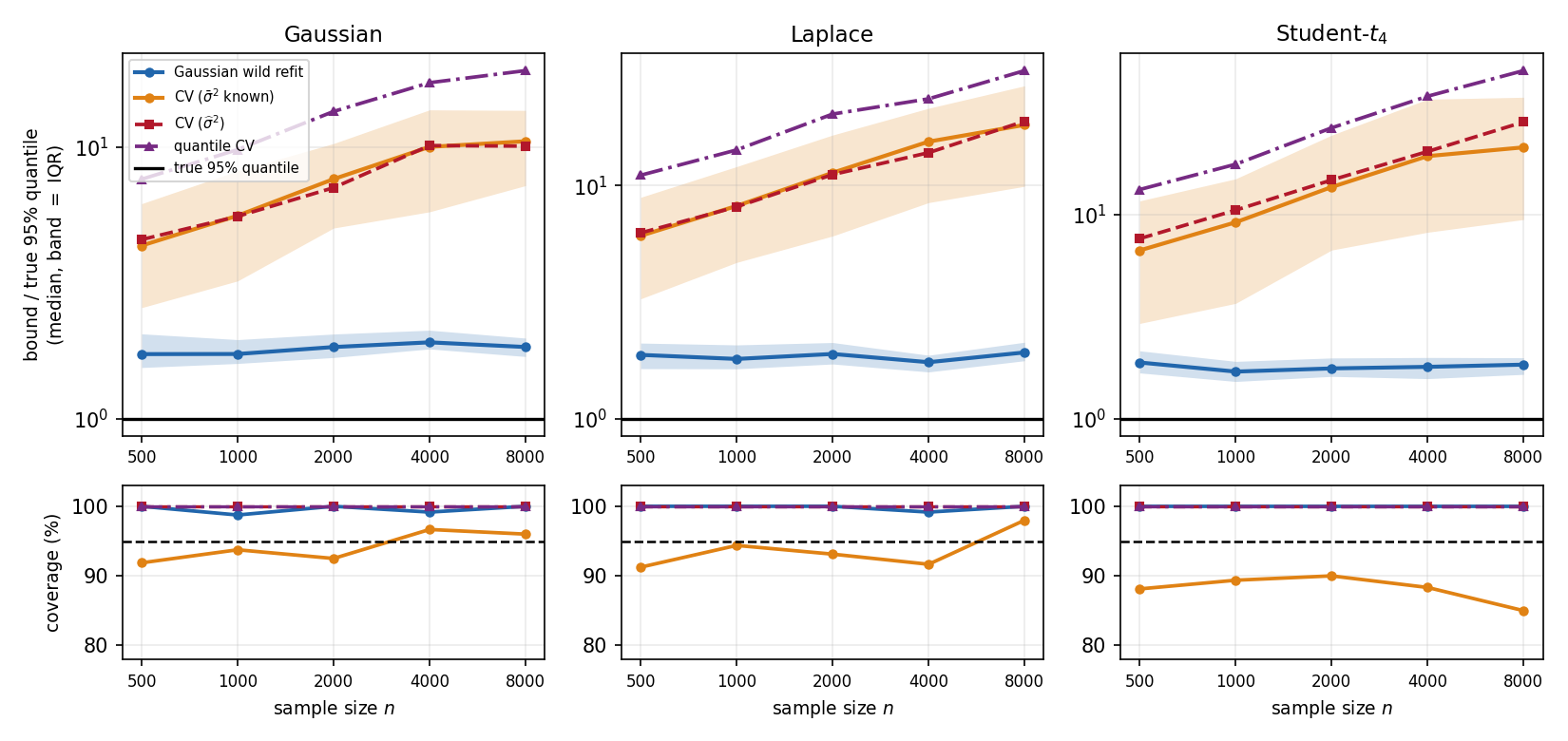}
\caption{The Gaussian refit tracks the target while cross-validation drifts away. \emph{Top:} the bound as a multiple of the true $95\%$ prediction-error quantile (median; bands, where shown, the interquartile range); the line at one is the target. \emph{Bottom:} empirical coverage against the stated $95\%$ (dashed). Cross-validation appears in the three forms of Table~\ref{tab:sim}; the known and estimated noise levels nearly coincide in the top panels.}
\label{fig:sim}
\end{figure}

\paragraph{The strongest cross-validation baseline.}
Against the nested cross-validation of \citet{batesHastieTibshirani2024cv} the picture is unchanged: on the small-sample cells it is close to the plain hold-out bound, at $4.5$--$12\times$ the target quantile with coverage $0.88$--$0.97$, so the Gaussian refit at $1.7$--$1.9\times$ remains several times tighter against the best cross-validation offers.

\paragraph{Repairing the interval does not help.}
Beyond the estimated noise level, the remaining repair abandons the normal approximation altogether: the quantile-CV column of Table~\ref{tab:sim}, an order statistic over $59$ independent holdout splits, the analogue of our calibration and free of moment assumptions.
It too behaves as Proposition~\ref{prop:floor} predicts: coverage is full on every cell, but the bound widens and inherits the divergence, $7.6\times$, $13\times$, $19\times$ across $n=500,2000,8000$ under heteroscedastic Gaussian noise, $13\times$ to $51\times$ under Student-$t_4$, and $547\times$ under Cauchy on the scale of \S\ref{sec:sim-moment}.
Across all three repairs the pattern is the same: the level of the interval is repairable, its noise-scale margin is not.

\paragraph{Portability to a nonlinear smoother.}
The movement is read by querying the fitted procedure, so it applies beyond the linear smoother. We test this on the constrained fit \eqref{eq:constrained}: the movement $m(a)$ is queried from $\check f$, while the envelope's leverage term and the bias input $b$, which have no closed form for a nonlinear fit, are computed from a linear kernel-ridge stand-in at the same penalty. Table~\ref{tab:port} reports the result: the bound holds coverage $0.99$--$1.00$ at $1.4$--$2.2\times$ the target quantile across $n\le4000$, while cross-validation runs $4.5$--$17\times$ and under-covers on half the cells; at $n=8000$ the heteroscedastic cell loosens to $3.0\times$, still covering. The guarantee of \S\ref{sec:theory} does not cover this estimator, so the coverage is empirical.

\begin{table}[t]
\centering
\caption{Portability to the constrained fit \eqref{eq:constrained}: the movement is queried from the nonlinear refit, the envelope and bias inputs from a linear kernel-ridge stand-in; accuracy and coverage as in \eqref{eq:metrics}. Cross-validation appears in its known-level form, since the estimated-level forms require degrees of freedom with no closed form for this fit. The tightest method with coverage at least $0.95$ is in bold.}
\label{tab:port}
\setstretch{1}\footnotesize
\setlength{\tabcolsep}{4pt}
\begin{tabular}{@{}llcccc@{}}
\toprule
noise & & $n{=}500$ & $1000$ & $2000$ & $4000$ \\
\midrule
Gaussian & ours & $\bm{1.76\ (0.99)}$ & $\bm{1.84\ (1.00)}$ & $\bm{1.85\ (0.99)}$ & $\bm{2.24\ (1.00)}$ \\
 & CV ($\bar\sigma^2$) & $4.49\ (0.96)$ & $5.77\ (0.95)$ & $6.76\ (0.95)$ & $11.8\ (0.98)$ \\
 & SURE & $0.30\ (0.46)$ & $0.65\ (0.51)$ & $0.09\ (0.45)$ & $0.34\ (0.47)$ \\
 & Rademacher & $0.27\ (0.41)$ & $0.22\ (0.38)$ & $0.29\ (0.42)$ & $0.48\ (0.48)$ \\
\addlinespace
Laplace & ours & $\bm{1.66\ (1.00)}$ & $\bm{1.71\ (1.00)}$ & $\bm{1.82\ (1.00)}$ & $\bm{1.78\ (1.00)}$ \\
 & CV ($\bar\sigma^2$) & $6.53\ (0.96)$ & $7.60\ (0.88)$ & $11.7\ (0.93)$ & $13.4\ (0.97)$ \\
 & SURE & $0.69\ (0.51)$ & $-0.16\ (0.39)$ & $0.01\ (0.46)$ & $1.36\ (0.53)$ \\
 & Rademacher & $0.29\ (0.31)$ & $0.27\ (0.33)$ & $0.30\ (0.33)$ & $0.37\ (0.40)$ \\
\addlinespace
Student-$t_4$ & ours & $\bm{1.75\ (1.00)}$ & $\bm{1.43\ (0.99)}$ & $\bm{1.74\ (1.00)}$ & $\bm{1.58\ (1.00)}$ \\
 & CV ($\bar\sigma^2$) & $6.49\ (0.89)$ & $6.97\ (0.88)$ & $12.0\ (0.91)$ & $17.1\ (0.91)$ \\
 & SURE & $-0.25\ (0.41)$ & $-0.81\ (0.39)$ & $-0.71\ (0.41)$ & $-0.80\ (0.43)$ \\
 & Rademacher & $0.31\ (0.34)$ & $0.30\ (0.36)$ & $0.27\ (0.36)$ & $0.21\ (0.29)$ \\
\bottomrule
\end{tabular}
\end{table}

\paragraph{Scope.}
The boundary of the method is charted in \appref{app:sim-scope}. As eigendecay slows the bound stays valid at full coverage but loses its margin: on the Mat\'ern-$3/2$ kernel it holds $2.7$--$3.3\times$, on the Mat\'ern-$1/2$ kernel $5$--$9\times$, no longer tighter than cross-validation at moderate $n$, the transition Proposition~\ref{prop:cv} anticipates through its exponent. Under a misspecified target a flexible fit stays covered, the residuals carrying the bias into the envelope, but a rigid over-smoothed fit is bias-dominated and coverage fails, the regime in which cross-validation is the appropriate tool.

\subsection{Robustness without moment assumptions}
\label{sec:sim-moment}

A stress test outside the main comparison: under standard Cauchy noise the variance is infinite and the cross-validation standard error estimates a quantity that does not exist. Measured against each dataset's own realized error, the Cauchy row of Table~\ref{tab:sim}, the Gaussian refit reads $4.6\times$ while cross-validation reads over $500\times$ with the noise level known and $286\times$ with it estimated, all at full coverage; the order-statistic repair of \S\ref{sec:sim-main} reads $547\times$. The bound depends only on the conditional symmetry of the signs and requires no moment of the noise.

\section{Real data}
\label{sec:realdata}

We test the certified bound on a real spatial field. Kernel ridge is at home on a smooth signal sampled over a domain, so we take as the regression function a real digital elevation model: the Jacksboro fault-zone terrain tile,\footnote{A U.S.\ Geological Survey surface-elevation grid, distributed as sample data with Matplotlib.} a $344\times403$ grid of surface heights.

\paragraph{Design.}
Lightly smoothed, the elevation $f^*$ lies in the radial-basis reproducing-kernel Hilbert space, so the fit is well specified; the design is a fixed random subsample of $n\in\{500,\,2000,\,8000\}$ sites in $[0,1]^2$, and the response $y_i=f^*(x_i)+w_i$ carries noise of the laws of \S\ref{sec:simulation}. Because the field is known, the excess risk $\mathcal{E}(\widehat f)=\norm{\widehat f-f^*}_n^2$ is measured exactly, and its $95\%$ quantile over $R$ noise draws is the yardstick \eqref{eq:metrics}, as in the synthetic study; the estimators and comparators are unchanged. The only synthetic ingredient is the noise, as in image denoising.

\begin{table}[t]
\centering
\caption{Real elevation field: accuracy and coverage \eqref{eq:metrics} on the digital elevation model, radial basis kernel, $R$ noise draws; the Cauchy row reports the median of bound over realized error (\S\ref{sec:sim-moment}). Cross-validation appears in the three forms of Table~\ref{tab:sim}; SURE and the Rademacher statistic are point estimates. The tightest method with coverage at least $0.95$ is in bold.}
\label{tab:realdata}
\setstretch{1}\footnotesize
\setlength{\tabcolsep}{2.6pt}
\begin{tabular}{@{}llcccccc@{}}
\toprule
noise & $n$ & Gaussian refit & CV ($\bar\sigma^2$) & CV ($\widehat\sigma^2$) & quantile CV & SURE & Rademacher \\
\midrule
Gaussian & $500$ & $3.03\ (1.00)$ & $2.05\ (0.94)$ & $\bm{2.09\ (1.00)}$ & $3.05\ (1.00)$ & $0.63\ (0.42)$ & $0.10\ (0.06)$ \\
Gaussian & $2000$ & $\bm{2.81\ (1.00)}$ & $3.40\ (0.94)$ & $3.04\ (1.00)$ & $4.75\ (1.00)$ & $0.96\ (0.56)$ & $0.36\ (0.15)$ \\
Gaussian & $8000$ & $\bm{2.55\ (1.00)}$ & $4.48\ (0.98)$ & $4.27\ (1.00)$ & $7.11\ (1.00)$ & $0.52\ (0.45)$ & $0.38\ (0.13)$ \\
\addlinespace
Laplace & $500$ & $3.37\ (1.00)$ & $\bm{2.81\ (0.95)}$ & $2.84\ (1.00)$ & $4.42\ (1.00)$ & $0.70\ (0.46)$ & $0.11\ (0.04)$ \\
Laplace & $2000$ & $\bm{2.75\ (1.00)}$ & $4.27\ (0.95)$ & $4.22\ (1.00)$ & $6.82\ (1.00)$ & $0.67\ (0.47)$ & $0.36\ (0.16)$ \\
Laplace & $8000$ & $\bm{2.43\ (1.00)}$ & $6.29\ (0.97)$ & $6.26\ (1.00)$ & $10.2\ (1.00)$ & $0.46\ (0.47)$ & $0.45\ (0.14)$ \\
\addlinespace
Student-$t_4$ & $500$ & $3.27\ (1.00)$ & $2.91\ (0.94)$ & $\bm{3.21\ (1.00)}$ & $5.03\ (1.00)$ & $0.55\ (0.40)$ & $0.16\ (0.05)$ \\
Student-$t_4$ & $2000$ & $\bm{2.72\ (1.00)}$ & $4.88\ (0.91)$ & $5.36\ (1.00)$ & $8.96\ (1.00)$ & $0.38\ (0.44)$ & $0.35\ (0.16)$ \\
Student-$t_4$ & $8000$ & $\bm{2.56\ (1.00)}$ & $7.59\ (0.86)$ & $9.69\ (1.00)$ & $17.2\ (1.00)$ & $-0.65\ (0.39)$ & $0.37\ (0.13)$ \\
\addlinespace
Cauchy & $2000$ & $\bm{6.09\ (1.00)}$ & $123\ (1.00)$ & $70.9\ (1.00)$ & $131\ (1.00)$ & $49.8\ (1.00)$ & $0.46\ (0.22)$ \\
\bottomrule
\end{tabular}
\end{table}

Table~\ref{tab:realdata} tells the same story as the synthetic design, now on a real field. The Gaussian refit holds full coverage on every cell, at $2.4$--$3.4\times$ the true quantile, and tightens as $n$ grows. Cross-validation's margin grows with $n$ in both forms, from $2.1$--$3.2\times$ at $n=500$ to $4.3$--$9.7\times$ at $n=8000$, the divergence of Propositions~\ref{prop:cv} and~\ref{prop:floor} on real data; with the noise level known, coverage also erodes to $0.86$ under Student-$t_4$ at $n=8000$, while the estimated level self-centres and covers at the same diverging margin. At $n=500$ cross-validation is the tighter valid bound on three of four noise laws; from $n=2000$ onward the Gaussian refit is tighter on every row, and the gap grows with $n$. Under Cauchy noise, measured against realized error, cross-validation reads $123\times$ known and $71\times$ estimated against the bound's $6.1\times$; every statistic is inflated by the realized noise, and the bound's single-digit factor is the informative one. Elsewhere SURE and the Rademacher statistic are point estimates and do not cover.

\section{Discussion}
\label{sec:discussion}

Extending the Rademacher wild refit to a Gaussian refit makes two tools available: Anderson's inequality licenses the replacement of the unobservable noise by a computable envelope, and the refit statistic acquires a known weighted-chi-square law, so an order statistic over repeated refits supplies the unobservable terms of a proven bound, one that holds its stated level, contracts at the minimax rate, and requires no moment of the noise.

The conditions delimiting the guarantee are those listed at the end of \S\ref{sec:intro}: validity is conservative and contingent on delocalized noise, the worst-case envelope is loose on slowly decaying spectra, and the proofs cover the well-specified regime. In practice the data-driven envelope stays within a small constant of the truth on the synthetic and real designs (\S\ref{sec:simulation}, \S\ref{sec:realdata}) and loosens as the eigendecay slows (\appref{app:sim-scope}), though its substitutions are estimates rather than upper bounds; it is robust to mild misspecification, since the residuals carry the bias into the envelope, and degrades only when the error becomes bias-dominated, where cross-validation is the appropriate tool. Extending the calibration to two-sample kernel losses is left to future work.

\section*{Data and code availability}
The elevation data of \S\ref{sec:realdata} are the U.S.\ Geological Survey Jacksboro fault-zone sample grid distributed with Matplotlib; all other data are simulated by the accompanying code. Python code reproducing the tables and figures, with fixed random seeds, is available from the authors; a DOI-stamped public release will accompany the journal version.

\section*{Acknowledgements}
Ni is supported by the Robert Goodell Ph.D.\ Student Fellowship for Research Excellence at Georgia Institute of Technology. Huo is partially supported by a subcontract of NSF grant 2229876, the A.\ Russell Chandler III Professorship at Georgia Institute of Technology, an NIH-sponsored Georgia Clinical \& Translational Science Alliance, and the Georgia Department of Transportation. The authors used a large language model (Anthropic's Claude) to assist with editing the manuscript text and developing the reproduction code; all technical content, proofs, and final wording were verified by the authors.

\appendix
\section{Proofs}
\label{app:proofs}

Throughout, the setting is that of \S\ref{sec:setup}: fixed design, $y=f^*+w$ with
conditional symmetry (given $|w|$, the sign vector $\varepsilon$ is uniform on
$\{-1,1\}^n$ and thus has independent coordinates), $f^*\in\HH$ with
$\norm{f^*}_\HH\le B$ and $\sup_x k(x,x)\le\kappa^2$, and the kernel ridge smoother
$H=K(K+n\lambda I)^{-1}$ of \eqref{eq:krr}, which is symmetric positive
semidefinite with $\norm{H}_{\mathrm{op}}\le1$, eigenvalues $h_j:=\mu_j/(\mu_j+\lambda)$,
and effective dimension $d_n=\operatorname{tr}(H)=\sum_j h_j$. We write $g:=\widehat f-f^*$,
so $g=(H-I)f^*+Hw$ and $\mathcal{E}(\widehat f)=\norm{g}_n^2$, and $\widetilde w=y-\widehat f$
for the residuals. For symmetric $H$ we use $H^\top H=H^2$ and
$(H^\top H)_{jj}=\sum_i H_{ij}^2=(H^2)_{jj}$ interchangeably.

The rate results of \S\ref{sec:theory} use, in addition to Assumption~\ref{as:decay},
the standing bounded-leverage regularity condition
\begin{equation}
\text{(D)}\qquad \max_i (H^\top H)_{ii}\;\le\; C_D\, d_n/n
\label{eq:leverage}
\end{equation}
for a constant $C_D$. Condition~\eqref{eq:leverage} holds whenever the kernel
eigenfunctions are uniformly bounded, since then
$(H^\top H)_{ii}=\sum_j h_j^2\varphi_j(x_i)^2\le C_\varphi^2\sum_j h_j^2\asymp C_\varphi^2 d_n$
up to the $1/n$ normalization; it is the standard delocalization condition of the
smoothing-spline literature and is stated as part of Assumption~\ref{as:decay}.

\subsection{Probabilistic tools}
\label{app:prob-tools}

\begin{lemma}[Anderson's inequality; \citealp{anderson1955integral}]
\label{lem:anderson-tool}
Let $X_1\sim\mathcal N(0,\Sigma_1)$ and $X_2\sim\mathcal N(0,\Sigma_2)$ with
$\Sigma_1\preceq\Sigma_2$ in the positive-semidefinite order. For every convex set
$C\subset\R^d$ that is symmetric about the origin, $\Pr(X_1\in C)\ge\Pr(X_2\in C)$.
\end{lemma}

\begin{lemma}[Hanson--Wright; \citealp{rudelsonVershynin2013hansonwright}]
\label{lem:hw}
Let $\varepsilon=(\varepsilon_1,\dots,\varepsilon_n)$ have independent, mean-zero,
unit sub-Gaussian coordinates and let $A$ be an $n\times n$ matrix. Then for every $t>0$,
\[
\Pr\bigl(|\varepsilon^\top A\varepsilon - \E\,\varepsilon^\top A\varepsilon| > t\bigr)
\;\le\; 2\exp\!\Bigl(-c\,\min\bigl\{t^2/\norm{A}_F^2,\; t/\norm{A}_{\mathrm{op}}\bigr\}\Bigr)
\]
for a universal constant $c>0$.
\end{lemma}

\begin{lemma}[Berry--Esseen for independent summands, Lyapunov form]
\label{lem:be}
Let $S:=\sum_{j}Y_j$ with $Y_j$ independent, mean zero, $\sum_j\Var(Y_j)=1$. Then
$\sup_t|\Pr(S\le t)-\Phi(t)|\le C_{\mathrm{BE}}\sum_j\E|Y_j|^3$ for a universal
constant $C_{\mathrm{BE}}$.
\end{lemma}

\begin{lemma}[Fourth-moment Berry--Esseen for a degree-two Rademacher chaos;
{\citealp{nourdinPeccatiReinert2010,doeblerPeccati2017}}]
\label{lem:chaos}
Let $N$ be a symmetric $n\times n$ matrix with zero diagonal and let
$F:=\varepsilon^\top N\varepsilon/(\sqrt2\,\norm{N}_F)$ for $\varepsilon$ a Rademacher
vector, so $\E F=0$ and $\Var F=1$. Then
$\sup_t|\Pr(F\le t)-\Phi(t)|\le C_{\mathrm{ch}}\,\norm{N}_{\mathrm{op}}/\norm{N}_F$
for a universal constant $C_{\mathrm{ch}}$.
\end{lemma}

Lemma~\ref{lem:chaos} is the specialization to a symmetric second-order Rademacher
form of the fourth-moment theorem: the total-variation (hence Kolmogorov) distance to
the normal is controlled by the maximal-influence and third/fourth standardized
cumulants of the chaos, and each of these is $O(\norm{N}_{\mathrm{op}}/\norm{N}_F)$.
Indeed $\kappa_3(\varepsilon^\top N\varepsilon)=8\operatorname{tr}(N^3)$ and
$\kappa_4=48\operatorname{tr}(N^4)-96\sum_i(N^2)_{ii}^2+32\sum_{ij}N_{ij}^4$, and each
invariant obeys $|\operatorname{tr}(N^m)|\le\norm{N}_{\mathrm{op}}^{m-2}\norm{N}_F^2$ for
$m\ge2$, so the standardized third and fourth cumulants are
$O(\norm{N}_{\mathrm{op}}/\norm{N}_F)$ and $O(\norm{N}_{\mathrm{op}}^2/\norm{N}_F^2)$
respectively, while the maximal influence is
$\max_i(N^2)_{ii}/\norm{N}_F^2\le\norm{N}_{\mathrm{op}}^2/\norm{N}_F^2$.

\subsection{Proof of Lemma~\ref{lem:anderson} (envelope domination)}
\label{app:proof-anderson}

Let $u,v\in\R^n$ be fixed with $u_i\ge v_i\ge0$ for every $i$, and $\xi\sim\mathcal N(0,I_n)$. Then
$A(\xi\circ u)\sim\mathcal N(0,\,A\operatorname{diag}(u^2)A^\top)$ and
$A(\xi\circ v)\sim\mathcal N(0,\,A\operatorname{diag}(v^2)A^\top)$. Because
$u_i\ge v_i\ge0$, the diagonal matrix $\operatorname{diag}(u^2-v^2)$ is positive
semidefinite, so
\[
A\operatorname{diag}(u^2)A^\top - A\operatorname{diag}(v^2)A^\top
= A\operatorname{diag}(u^2-v^2)A^\top \succeq 0
\]
by congruence, i.e.\ $A\operatorname{diag}(v^2)A^\top\preceq A\operatorname{diag}(u^2)A^\top$.
The set $C_t:=\{z:\norm{z}\le t\}$ is convex and symmetric about the origin. Applying
Lemma~\ref{lem:anderson-tool} with $\Sigma_1=A\operatorname{diag}(v^2)A^\top$,
$\Sigma_2=A\operatorname{diag}(u^2)A^\top$ and $C=C_t$ gives
$\Pr(\norm{A(\xi\circ v)}\le t)\ge\Pr(\norm{A(\xi\circ u)}\le t)$; taking complements
yields the claim. \qed

\medskip\noindent\emph{Nonnegativity is necessary.} With $u=(1,0)$, $v=(-3,0)$ the
matrix $\operatorname{diag}(u^2-v^2)=\operatorname{diag}(-8,0)$ is not positive
semidefinite and the conclusion can fail; the hypothesis $u_i\ge v_i\ge0$ is used exactly
here.

\medskip\noindent\emph{Pathwise application (used in \S\ref{app:proof-validity}).} When
the envelope $a=a(\varepsilon)$ is data-dependent, Lemma~\ref{lem:anderson} is applied
conditionally on $(\varepsilon,|w|)$: for each fixed realization, $a(\varepsilon)$ is a
fixed vector with $a_i(\varepsilon)\ge|w_i|$ for every $i$, and $\xi\perp(\varepsilon,|w|)$,
so the inequality holds verbatim and the resulting quantile of $\norm{H(\xi\circ|w|)}_n$
is a function of $|w|$ alone.

\medskip\noindent\emph{No Rademacher analogue.} For a two-point multiplier the statement
is false: with $u=(1,1)$, $v=(1,0)$, $A=[1,1]$ one has
$\Pr(|\varepsilon_1+\varepsilon_2|>\tfrac12)=\tfrac12<1=\Pr(|\varepsilon_1|>\tfrac12)$.
The two-point law is not infinitely divisible, so the variance-increment decomposition
underlying Lemma~\ref{lem:anderson-tool} has no two-point counterpart.

\subsection{The wild-refit propositions and the degeneracy \texorpdfstring{(Corollary~\ref{cor:w25krr})}{}}
\label{app:degeneracy}

We first restate the two results of \citet{wainwright2025wild} that Corollary~\ref{cor:w25krr} specializes, stated with the optimism $\Opt^\star(\widehat f):=\tfrac1n\sum_i w_i\bigl(\widehat f(x_i)-f^*(x_i)\bigr)$ and the wild refit and wild optimism of \eqref{eq:wild-refit}.

\begin{proposition}[Wild-refit inequality; \citealp{wainwright2025wild}]
\label{prop:w25}
Under Assumption~\ref{as:noise}, for a radius $r\ge\norm{\widehat f-f^\dagger}_n$ with $f^\dagger$ of \eqref{eq:err-split}, and the scale $\rho$ at which the wild refit \eqref{eq:wild-refit} satisfies $\norm{\widehat f^{\bullet}_\rho-\widehat f}_n=2r$, the optimism obeys, for any $t>0$ and with probability at least $1-4\exp(-t^2)$,
\[
\Opt^\star(\widehat f) \;\le\; \Optw^{\bullet}_\rho \;+\; \bigl(3r+\norm{f^\dagger-f^*}_n\bigr)\frac{2\norm{w}_\infty\,t}{\sqrt n} \;+\; A_n ,
\]
where $\norm{w}_\infty:=\max_i|w_i|$ and the pilot-approximation term $A_n$ is
\[
A_n \;:=\; \sup_{\norm{f-\widehat f}_n\le 2r}\ \tfrac1n\textstyle\sum_i \varepsilon_i\bigl(\widehat f(x_i)-f^*(x_i)\bigr)\bigl(f(x_i)-\widehat f(x_i)\bigr) .
\]
\end{proposition}

The rate is a second, separate ingredient, obtained from the same cross term over the fit's neighbourhood: maximized over an $r$-ball in the wild form, it is the wild complexity
\begin{equation}
W_n(r) \;:=\; \sup_{\norm{f-\widehat f}_n\le r}\ \tfrac1n\textstyle\sum_i\varepsilon_i\widetilde w_i\bigl(f(x_i)-\widehat f(x_i)\bigr).
\label{eq:wildcomplexity-sm}
\end{equation}

\begin{proposition}[Wild-refit rate; \citealp{wainwright2025wild}]
\label{prop:w25rate}
Let $r^*$ be the critical radius, the fixed point $W_n(r^*)\asymp(r^*)^2$ of the wild complexity \eqref{eq:wildcomplexity-sm}. With probability at least $1-4\exp(-t^2)$, and up to the remainder of Proposition~\ref{prop:w25}, the estimation error obeys $\norm{\widehat f-f^\dagger}_n\lesssim r^*$, so by \eqref{eq:err-split}, $\mathcal{E}(\widehat f)\lesssim(r^*)^2+\norm{f^\dagger-f^*}_n^2$. A wild complexity growing sublinearly in $r$ has a critical radius that shrinks with $n$; a linear one leaves $r^*$ fixed.
\end{proposition}

For part (i) of the corollary, the wild optimism at the linear smoother is the explicit quadratic form $\Optw^{\bullet}_\rho=\tfrac{\rho}{n}(\varepsilon\circ\widetilde w)^\top H(\varepsilon\circ\widetilde w)-\tfrac1n(\varepsilon\circ\widetilde w)^\top H\widetilde w$, computable from the data and the drawn signs, while every term of the remainder of Proposition~\ref{prop:w25} involves $w$ or $f^*$; the display \eqref{eq:w25krr} is Proposition~\ref{prop:w25} with the kernel-ridge bias $\norm{f^\dagger-f^*}_n=\norm{(H-I)f^*}_n$ substituted.

For part (ii), for the linear smoother \eqref{eq:krr} with pilot equal to the fit, $\widetilde w=(I-H)y$, and the smoother places no constraint on $f$ beyond $\norm{f-\widehat f}_n\le r$. Writing $g=f-\widehat f$, the wild complexity \eqref{eq:wildcomplexity-sm} is a linear maximization,
\[
W_n(r)=\sup_{\norm{g}_n\le r}\langle\varepsilon\circ\widetilde w,\,g\rangle_n
=r\,\norm{\varepsilon\circ\widetilde w}_n=r\,\norm{\widetilde w}_n ,
\]
by Cauchy--Schwarz (attained at $g=r\,\varepsilon\circ\widetilde w/\norm{\varepsilon\circ\widetilde w}_n$) and $|\varepsilon_i|=1$. The critical radius of Proposition~\ref{prop:w25rate} solves $W_n(r^*)\asymp(r^*)^2$, that is $r^*\norm{\widetilde w}_n\asymp(r^*)^2$, so
\[
r^*\asymp\norm{\widetilde w}_n=\norm{(I-H)y}_n .
\]
The residual norm $\norm{(I-H)y}_n$ does not vanish as $n\to\infty$ (the residuals carry the noise), so $r^*$ is of constant order. By \eqref{eq:err-split} and $\norm{\widehat f-f^\dagger}_n\lesssim r^*$,
\[
\norm{\widehat f-f^*}_n\le\norm{\widehat f-f^\dagger}_n+\norm{f^\dagger-f^*}_n\lesssim r^*+\norm{(H-I)f^*}_n\asymp\norm{(I-H)y}_n ,
\]
while $\norm{\widehat f-f^*}_n=O_P(n^{-s/(2s+1)})$ by Lemma~\ref{lem:krr-rate}. The wild-refit bound thus exceeds the truth by a factor of order at least $\norm{(I-H)y}_n/n^{-s/(2s+1)}\asymp n^{s/(2s+1)}\to\infty$.

\subsection{The error decomposition \texorpdfstring{\eqref{eq:wr-ineq}}{}}
\label{app:wr-inequality}

For the linear smoother $\widehat f=Hy=Hf^*+Hw$, the triangle inequality in the empirical norm gives
\[
\norm{\widehat f-f^*}_n=\norm{(H-I)f^*+Hw}_n\le\norm{Hw}_n+\norm{(H-I)f^*}_n .
\]
The bias $\norm{(H-I)f^*}_n$ is bounded over the reproducing-kernel-ball by $b$ in
\appref{app:envelope}, giving \eqref{eq:wr-ineq}. This is a deterministic bound. The completion and calibration of \S\ref{sec:completion} bound the noise term $\norm{Hw}_n$ by a computable refit quantile, with the domination established probabilistically in \appref{app:proof-validity}. For a general firmly non-expansive $\mathcal{M}$ the linear
decomposition is replaced by the contraction argument of \citet{wainwright2025wild}, which
yields the same bound with $\norm{Hw}_n$ the norm of the procedure's response to the noise.

\subsection{The exact refit law}
\label{app:exact-law}

For the linear smoother, the squared refit movement at envelope $a$ is the quadratic
form
\[
m(a)^2=\norm{H(\xi\circ a)}_n^2=\tfrac1n\,\xi^\top\operatorname{diag}(a)H^\top H\operatorname{diag}(a)\,\xi
=\tfrac1n\sum_{j=1}^n\ell_j(a)\,Z_j^2,
\]
where $\ell_1(a)\ge\cdots\ge\ell_n(a)\ge0$ are the eigenvalues of
$\operatorname{diag}(a)H^\top H\operatorname{diag}(a)$ and $Z_j$ are independent standard
normals (spectral decomposition of the positive-semidefinite form). Thus $n\,m(a)^2$ is a
weighted sum of independent $\chi^2_1$ variables, whose distribution function is available
in closed form by numerical inversion of the characteristic function
\mbox{$t\mapsto\prod_j(1-2it\ell_j(a)/n)^{-1/2}$} \citep{imhof1961computing}. The order statistic
\eqref{eq:orderstat} estimates the quantiles of this law directly, which is why the
Gaussian multiplier makes the calibration step exact up to Monte Carlo error.

\subsection{The worst-case envelope}
\label{app:envelope}

Under Assumption~\ref{as:model} we construct a data-measurable envelope $a$ with $a_i\ge|w_i|$ for every $i$ and bias input
$b\ge\norm{(H-I)f^*}_n$, so that Theorem~\ref{thm:validity} certifies the resulting
$\widehat U_\alpha$.

\paragraph{Envelope.} The residual is $\widetilde w_i=y_i-\widehat f(x_i)=w_i-g_i$, so
$|w_i|\le|\widetilde w_i|+|g_i|$. By the reproducing property and Cauchy--Schwarz,
\[
|g_i|=|\langle \widehat f-f^*,\,k(x_i,\cdot)\rangle_\HH|
\le \norm{\widehat f-f^*}_\HH\,\sqrt{k(x_i,x_i)}
\le \bigl(\norm{\widehat f}_\HH+B\bigr)\sqrt{k(x_i,x_i)}=:S_i,
\]
where $\norm{\widehat f}_\HH$ is computable from the fit and $\sqrt{k(x_i,x_i)}\le\kappa$.
Setting $a_i:=|\widetilde w_i|+S_i$ gives $a_i\ge|w_i|$ pointwise, as required by
Lemma~\ref{lem:anderson}.

\paragraph{Bias input.} In the eigenbasis of $H$, $(I-H)$ has eigenvalues
$\lambda/(\mu_j+\lambda)$, so with $\langle f^*,\phi_j\rangle$ the coordinates of $f^*$ and
$\norm{f^*}_\HH^2=\sum_j\langle f^*,\phi_j\rangle^2/\mu_j\le B^2$,
\[
\norm{(H-I)f^*}_n^2=\sum_j\Bigl(\tfrac{\lambda}{\mu_j+\lambda}\Bigr)^2\langle f^*,\phi_j\rangle^2
=\sum_j\frac{\lambda^2\mu_j}{(\mu_j+\lambda)^2}\cdot\frac{\langle f^*,\phi_j\rangle^2}{\mu_j}
\le \sup_{\mu\ge0}\frac{\lambda^2\mu}{(\mu+\lambda)^2}\,B^2=\frac{\lambda}{4}B^2,
\]
the supremum being attained at $\mu=\lambda$. Hence $b:=\tfrac12 B\sqrt{\lambda}$ satisfies
$b\ge\norm{(H-I)f^*}_n$. Both $a$ and $b$ are computable given the known radius $B$ of
Assumption~\ref{as:model}, and with these choices every term of $\widehat U_\alpha$ is a
valid upper bound.

\subsection{Proof of Theorem~\ref{thm:validity} (validity)}
\label{app:proof-validity}

Fix the magnitudes $|w|$ and recall $M=\operatorname{diag}(|w|)H^\top H\operatorname{diag}(|w|)$,
$N=M-\operatorname{diag}(M)$, and the functionals \eqref{eq:deloc}. Write
$\mu_M:=\operatorname{tr}(M)/n$, $\sigma_M:=\sqrt2\,\norm{M}_F/n$, and standardize the two
quadratic forms
\[
G:=\frac{\xi^\top M\xi/n-\mu_M}{\sigma_M},\qquad
R:=\frac{\varepsilon^\top M\varepsilon/n-\mu_M}{\sigma_M},
\]
so $\E G=\E R=0$, $\Var G=1$, and $\Var R=1-\rho^2$ (the diagonal of $M$ contributes no
fluctuation to the Rademacher form because $\varepsilon_i^2\equiv1$).

\emph{Step 1 (reduction).} On the event $\{\mathcal{E}(\widehat f)>\widehat U_\alpha\}$,
$\sqrt{\mathcal{E}(\widehat f)}=\norm{g}_n\le\norm{(H-I)f^*}_n+\norm{H(\varepsilon\circ|w|)}_n\le b+T$
with $T:=\norm{H(\varepsilon\circ|w|)}_n$, using $b\ge\norm{(H-I)f^*}_n$, while
$\sqrt{\widehat U_\alpha}=q_{1-\alpha}(a)+b$. By Lemma~\ref{lem:anderson} applied
pathwise in $\varepsilon$ (\S\ref{app:proof-anderson}), the population $(1-\alpha)$-quantile
of $\norm{H(\xi\circ a)}_n$ dominates that of $\norm{H(\xi\circ|w|)}_n=:Q_0$, a function
of $|w|$ alone; the order statistic \eqref{eq:orderstat} estimates the former with error
$O_P(L^{-1/2})$ by the Dvoretzky--Kiefer--Wolfowitz inequality. Hence
$q_{1-\alpha}(a)\ge Q_0$ up to $C_2/\sqrt L$, and
\[
\{\mathcal{E}(\widehat f)>\widehat U_\alpha\}\subseteq\{b+T>b+Q_0\}=\{T>Q_0\},
\]
so that
\[
\Pr(\mathcal{E}(\widehat f)>\widehat U_\alpha\mid|w|)\le\Pr(T>Q_0\mid|w|)+\tfrac{C_2}{\sqrt L}.
\]

\emph{Step 2 (Gaussian threshold).} Since $Q_0$ is the $(1-\alpha)$-quantile of
$\norm{H(\xi\circ|w|)}_n$ and $n\,T^2=\varepsilon^\top M\varepsilon$,
$n\,Q_0^2$ is the $(1-\alpha)$-quantile of $\xi^\top M\xi$, so
$\{T>Q_0\}=\{R>g_{1-\alpha}\}$ with $g_{1-\alpha}$ the $(1-\alpha)$-quantile of $G$. In
the eigenbasis of $M$, $\xi^\top M\xi=\sum_j\ell_j Z_j^2$ with $\ell_j\ge0$ (as $M\succeq0$),
so $G=\sum_j Y_j$ with $Y_j:=\ell_j(Z_j^2-1)/(\sqrt2\norm{M}_F)$ independent, mean zero,
$\sum_j\Var Y_j=1$, and $\sum_j\E|Y_j|^3\le C\,\lambda_{\max}(M)/\norm{M}_F=C\delta$ using
$\sum_j\ell_j^3\le\lambda_{\max}(M)\sum_j\ell_j^2$. Lemma~\ref{lem:be} gives
$\sup_t|\Pr(G\le t)-\Phi(t)|\le C_0\delta$, and quantile inversion (with $\phi$ bounded
below near $z_{1-\alpha}$) yields $g_{1-\alpha}\ge z_{1-\alpha}-C_0\delta/\phi(z_{1-\alpha})$.

\emph{Step 3 (Rademacher tail).} The centered part of $\varepsilon^\top M\varepsilon$ equals
$\varepsilon^\top N\varepsilon$, so $R$ has the law of $\sqrt{1-\rho^2}\,F$ with $F$ the
standardized chaos of Lemma~\ref{lem:chaos}. Therefore, for any threshold $s$,
$\Pr(R>s)\le 1-\Phi\bigl(s/\sqrt{1-\rho^2}\bigr)+C_1\delta_N$. Taking
$s=g_{1-\alpha}>0$ (as $\alpha<\tfrac12$) and using $\sqrt{1-\rho^2}\le1$ so
$g_{1-\alpha}/\sqrt{1-\rho^2}\ge g_{1-\alpha}\ge z_{1-\alpha}-C_0\delta/\phi(z_{1-\alpha})$,
a one-term Taylor bound on $1-\Phi$ gives
\[
\Pr(T>Q_0\mid|w|)=\Pr(R>g_{1-\alpha})
\le 1-\Phi(z_{1-\alpha})+\frac{C_0\delta}{\phi(z_{1-\alpha})}+C_1\delta_N
=\alpha+\frac{C_0\delta}{\phi(z_{1-\alpha})}+C_1\delta_N .
\]
Since $\phi(z_{1-\alpha})\le1$ we have $C_1\delta_N\le C_1\delta_N/\phi(z_{1-\alpha})$, so the
two delocalization terms combine into
$\Pr(T>Q_0\mid|w|)\le\alpha+(\max(C_0,C_1)/\phi(z_{1-\alpha}))(\delta+\delta_N)$. Adding the
Monte Carlo term $C_2/\sqrt L$ of Step 1 and relabelling the two universal constants as
$\max(C_0,C_1)\mapsto C_0$ and $C_2\mapsto C_1$ gives \eqref{eq:validity}.
Finally, the two-sided Weyl inequalities $\lambda_{\max}(N)\le\lambda_{\max}(M)\le\delta\norm{M}_F$
and $|\lambda_{\min}(N)|\le\max_i M_{ii}\le\rho\norm{M}_F$, together with
$\norm{N}_F=\norm{M}_F\sqrt{1-\rho^2}$, give
$\delta_N\le(\delta+\rho)/\sqrt{1-\rho^2}$, so the right-hand side of \eqref{eq:validity}
tends to $\alpha$ whenever $\delta\to0$, $\rho\to0$, and $L\to\infty$. \qed

\medskip\noindent\emph{Conservativeness.} The variance deficit $\Var R=1-\rho^2<1$ means
the true threshold statistic is less dispersed than its Gaussian surrogate, so the
Gaussian quantile is conservative; the numerically observed miscoverage is nonpositive
throughout the delocalized regime.

\subsection{Proof of Theorem~\ref{thm:rate} (rate optimality)}
\label{app:proof-rate}

We first record the risk order (used again in \S\ref{app:proof-cv}).

\begin{lemma}[Risk order]
\label{lem:risk}
Under Assumptions~\ref{as:model}, \ref{as:noise}, \ref{as:decay}, and~\ref{as:scale},
$\mathcal{E}(\widehat f)=O_P\bigl(n^{-2s/(2s+1)}\bigr)$. This is Lemma~\ref{lem:krr-rate}
of the main text.
\end{lemma}

\begin{proof}
$\mathcal{E}(\widehat f)=\norm{(H-I)f^*+Hw}_n^2\le2\norm{(H-I)f^*}_n^2+2\norm{Hw}_n^2$.
The bias term is at most $B^2\lambda/2$ by the computation in \S\ref{app:envelope}. For the
variance term, $\E[\norm{Hw}_n^2\mid|w|]=\tfrac1n\operatorname{tr}(H^2\Sigma_w)$ with
$\Sigma_w:=\operatorname{Cov}(w\mid|w|)$; under conditional symmetry
$\E[w_iw_j\mid|w|]=|w_i||w_j|\,\E[\varepsilon_i\varepsilon_j\mid|w|]=0$ for $i\ne j$, so
$\Sigma_w$ is diagonal with entries $|w_i|^2$; averaging over the magnitudes, the
sub-Gaussian tail of Assumption~\ref{as:scale} gives $\E w_i^2\le4\sigma^2$, so
$\E\norm{Hw}_n^2\le4\sigma^2\operatorname{tr}(H^2)/n\le4\sigma^2 d_n/n$. Both terms are
$O(n^{-2s/(2s+1)})$ since $\lambda\asymp n^{-2s/(2s+1)}$ and $d_n\asymp n^{1/(2s+1)}$;
Markov's inequality gives the claim. No independence is used; diagonal conditional
covariance suffices.
\end{proof}

\emph{Part (i): $\widehat U_\alpha=O_P(n^{-2s/(2s+1)})$.} With the worst-case envelope,
$\norm{a}_n^2=O_P(1)$: the residual energy satisfies
$\norm{\widetilde w}_n^2\le2\norm{w}_n^2+2\mathcal{E}(\widehat f)=O_P(\sigma^2)$ by
Assumption~\ref{as:scale} and Lemma~\ref{lem:risk}, and the summand
$S_i\le(\norm{\widehat f}_\HH+B)\kappa$ is $O_P(1)$ because
$\norm{\widehat f}_\HH\le\norm{Hf^*}_\HH+\norm{Hw}_\HH\le B+O_P\bigl(\sigma\sqrt{d_n/(n\lambda)}\bigr)$,
the first term by the variational characterization \eqref{eq:krr} (at noiseless input the
minimizer beats $f^*$ itself, so $\norm{Hf^*}_\HH\le\norm{f^*}_\HH\le B$), the second from
$\E\norm{Hw}_\HH^2\le4\sigma^2\sum_j\mu_j/\bigl(n(\mu_j+\lambda)^2\bigr)\le4\sigma^2 d_n/(n\lambda)$,
with $d_n\asymp n\lambda$ under Assumption~\ref{as:decay}. Also
$b^2=B^2\lambda/4=O(n^{-2s/(2s+1)})$. The refit movement has conditional mean
\[
\begin{aligned}
\E_\xi\norm{H(\xi\circ a)}_n^2
&=\tfrac1n\sum_j a_j^2(H^\top H)_{jj}
\le\bigl(\max_j(H^\top H)_{jj}\bigr)\norm{a}_n^2\\
&\le C_D\,\tfrac{d_n}{n}\,O_P(1)
=O_P\bigl(n^{-2s/(2s+1)}\bigr),
\end{aligned}
\]
using condition~\eqref{eq:leverage} and $d_n/n\asymp n^{-2s/(2s+1)}$. The map
$\xi\mapsto\norm{H(\xi\circ a)}_n$ is Lipschitz with constant
$\norm{a}_\infty\norm{H}_{\mathrm{op}}/\sqrt n\le\norm{a}_\infty/\sqrt n$, so by
Borell--TIS concentration its $(1-\alpha)$-quantile exceeds its mean by at most
$O_P(\norm{a}_\infty\sqrt{\log(1/\alpha)}/\sqrt n)=O_P(\sqrt{\log n}/\sqrt n)$, which is
$o(n^{-s/(2s+1)})$ for fixed $s$ since $n^{1/(2s+1)}$ dominates $\log n$. Squaring,
$\widehat U_\alpha=(q_{1-\alpha}(a)+b)^2=O_P(n^{-2s/(2s+1)})$, the minimax rate over
$\{f\in\HH:\norm{f}_\HH\le B\}$ \citep{stone1982optimal,yang1999information}.

\emph{Part (ii): $\widehat U_\alpha/\mathcal{E}(\widehat f)=O_P(1)$ under
Assumption~\ref{as:energy}.} It suffices to lower-bound $\mathcal{E}(\widehat f)$ by the same
rate. Write $\mathcal{E}(\widehat f)=\norm{Hw}_n^2+2\langle Hw,(H-I)f^*\rangle_n+\norm{(H-I)f^*}_n^2$.
Conditional on $|w|$, $\norm{Hw}_n^2=\tfrac1n\varepsilon^\top D H^2 D\varepsilon$ with
$D:=\operatorname{diag}(|w|)$, whose mean is
$\tfrac1n\sum_i|w_i|^2(H^\top H)_{ii}\ge c_0 d_n/n$ by Assumption~\ref{as:energy}. Since
$\norm{DH^2D}_{\mathrm{op}}\le\norm{w}_\infty^2=O_P(\sigma^2\log n)$ by the sub-Gaussian
tail of Assumption~\ref{as:scale} and
$\norm{DH^2D}_F^2\le\norm{DH^2D}_{\mathrm{op}}\operatorname{tr}(DH^2D)\asymp\norm{w}_\infty^2 c_0 d_n$,
Hanson--Wright (Lemma~\ref{lem:hw}) with $t=\tfrac12\operatorname{tr}(DH^2D)$ gives deviation
probability $\to0$ provided $d_n/\log n\to\infty$, which holds under
Assumption~\ref{as:decay}. Hence $\norm{Hw}_n^2\ge\tfrac{c_0}{2}d_n/n$ with high
probability. The cross term is a mean-zero Rademacher sum
$2\langle Hw,(H-I)f^*\rangle_n=\tfrac2n\sum_i\varepsilon_i|w_i|\,\eta_i$ with
$\eta:=H(H-I)f^*$; its conditional standard deviation is at most
$\tfrac2n\norm{w}_\infty\norm{\eta}_2\le\norm{w}_\infty B\sqrt{\lambda}/\sqrt n
=O_P(\sqrt{d_n\log n}/n)=o_P(d_n/n)$ since $d_n\gg\log n$, using
$\norm{\eta}_2^2\le\norm{(H-I)f^*}_2^2=n\norm{(H-I)f^*}_n^2\le nB^2\lambda/4$. Therefore
\[
\mathcal{E}(\widehat f)\ge\tfrac{c_0}{2}\tfrac{d_n}{n}-o_P\bigl(\tfrac{d_n}{n}\bigr)+0
\ge\tfrac{c_0}{4}\tfrac{d_n}{n}\quad\text{w.h.p.,}
\]
and combining with Part (i), $\widehat U_\alpha/\mathcal{E}(\widehat f)=O_P(1)$. \qed

\medskip\noindent\emph{Exponential-decay kernels.} If $\mu_j$ decays faster than any
polynomial (e.g.\ the Gaussian RBF), $d_n\asymp\log n$ and the Borell--TIS
$\sqrt{\log n/n}$ term is no longer dominated, so Part (i) carries an extra factor:
$\widehat U_\alpha=O_P(\log n/n)$ and $\widehat U_\alpha/\mathcal{E}(\widehat f)=O_P(\log n)$.
This is the case reported empirically in \S\ref{sec:simulation}.

\subsection{Proof of Proposition~\ref{prop:cv} (cross-validation margin)}
\label{app:proof-cv}

Consider the hold-out bound \eqref{eq:cvbound} with hold-out size $m\asymp n$ and losses
$\ell_i:=(y_i-\widehat f^{\mathrm{tr}}(x_i))^2$, $\widehat{\mathrm{se}}^2:=s_\ell^2/m$ with
$s_\ell^2$ the sample variance of $\{\ell_i\}$. Decompose $\ell_i=w_i^2+b_i$ with
$b_i:=2w_i\Delta_i+\Delta_i^2$ and $\Delta_i:=f^*(x_i)-\widehat f^{\mathrm{tr}}(x_i)$. By the
reverse triangle inequality for the empirical standard-deviation seminorm,
$s_\ell\ge s_{w^2}-s_b$. Assumption~\ref{as:spread} gives
$s_{w^2}^2=\operatorname{empvar}_i\{w_i^2\}\ge\kappa_0$. For the perturbation,
\[
s_b^2\le\operatorname{empmean}(b_i^2)\le 8\bigl(\max_i w_i^2\bigr)\norm{\Delta}_n^2
+2\norm{\Delta}_n^2\norm{\Delta}_\infty^2 .
\]
Under Assumption~\ref{as:scale}, $\max_i w_i^2=O_P(\sigma^2\log m)$, and
$\norm{\Delta}_\infty\le\kappa(\norm{\widehat f^{\mathrm{tr}}}_\HH+B)=O_P(1)$ by
Assumption~\ref{as:model} and the fitted-norm bound in the proof of
Theorem~\ref{thm:rate}; the fit-rate condition
$\norm{\Delta}_n^2\log m\to_P0$, delivered automatically by Lemma~\ref{lem:risk} since
$\norm{\Delta}_n^2=O_P(n^{-2s/(2s+1)})$, then gives $s_b^2\to_P0$. Hence for large $m$,
$s_\ell\ge\tfrac12\sqrt{\kappa_0}$ with high probability, so
$\widehat{\mathrm{se}}^2\ge\kappa_0/(4m)$ and
\[
z_{1-\alpha}\,\widehat{\mathrm{se}}\ge\frac{z_{1-\alpha}}{2}\frac{\sqrt{\kappa_0}}{\sqrt m}
=\Omega_P\bigl(n^{-1/2}\bigr),
\]
independent of the fit. Dividing by $\mathcal{E}(\widehat f)=O_P(n^{-2s/(2s+1)})$
(Lemma~\ref{lem:risk}) gives the exponent $2s/(2s+1)-1/2=(2s-1)/(2(2s+1))>0$ for $s>1/2$,
which is \eqref{eq:cvdiverge}. \qed

\medskip\noindent\emph{Degenerate case and rate-only claim.} If $w_i^2$ is constant across
observations, $\operatorname{empvar}_i\{w_i^2\}=0$, Assumption~\ref{as:spread} fails, and
the margin floor vanishes, the single case in which cross-validation keeps pace. The
constant in $\Omega_P$ depends on the finer structure of $\{w_i^2\}$, so the statement is a
rate-level, not constant-level, claim. Under $K$-fold rather than hold-out splitting the
fold losses are dependent and no unbiased standard-error estimator exists
\citep{bengioGrandvalet2004no}; the hold-out case is proved and the $K$-fold case matches
empirically.

\subsection{Proof of Proposition~\ref{prop:floor} (holdout floor)}
\label{app:proof-floor}

Throughout, condition on the training responses, so the profile $\delta=(\delta_i)_{i\in V}$ is
deterministic and the held-out losses are independent with
$\ell_i\sim\sigma^2\chi_1'^2(\lambda_i)$, $\lambda_i:=\delta_i^2/\sigma^2$. Write
$P_{\delta,\sigma}$ for the joint law of $(\ell_i)_{i\in V}$.

\medskip\noindent\emph{Step 1: ridge construction.} Fix a problem $(\delta^0,\sigma)$ in the
interior class of Proposition~\ref{prop:floor} and set $\varepsilon:=\kappa m^{-1/2}$. Define
the paired problem
\[
\sigma_1^2:=\sigma^2(1-\varepsilon),
\qquad
\delta_i^1:=\bigl(\,(\delta_i^0)^2+\sigma^2\varepsilon\,\bigr)^{1/2},
\]
so that $(\delta_i^1)^2+\sigma_1^2=(\delta_i^0)^2+\sigma^2$ for every $i$: the per-point
first moments of the losses are matched exactly. The target moves up by
$\norm{\delta^1}_m^2-\norm{\delta^0}_m^2=\sigma^2\varepsilon\ge\kappa\sigma_0^2m^{-1/2}$.
For $n$ large the paired problem lies in the validity class \eqref{eq:floor-class}:
$\norm{\delta^1}_m^2\le\tfrac12\sigma_0^2+\tfrac94\sigma_0^2\varepsilon\le\sigma_0^2$,
$\norm{\delta^1}_\infty\le(\tfrac14C_\infty^2\sigma_0^2+\tfrac94\sigma_0^2\varepsilon)^{1/2}
\le C_\infty\sigma_0$, and
$\sigma_1^2\ge\sigma_0^2(1-\varepsilon)\ge\tfrac14\sigma_0^2$.

\medskip\noindent\emph{Step 2: the loss laws are nearly indistinguishable along the ridge.}
The map $x\mapsto x^2$ erases signs, so the loss law lifts to a symmetric mixture: if
$S_i$ are independent Rademacher signs, then $(S_iZ_i)^2$ with
$Z_i\sim\mathcal N(\delta_i,\sigma^2)$ has law $\sigma^2\chi_1'^2(\lambda_i)$, and the same
holds with the mixture $M(\delta_i,\sigma):=\tfrac12\mathcal N(\delta_i,\sigma^2)
+\tfrac12\mathcal N(-\delta_i,\sigma^2)$ in place of $\mathcal N(\delta_i,\sigma^2)$. Since a
measurable map can only decrease total variation,
\[
\mathrm{TV}\bigl(P_{\delta^0,\sigma},\,P_{\delta^1,\sigma_1}\bigr)
\;\le\;
\mathrm{TV}\Bigl(\textstyle\bigotimes_{i\in V}M(\delta_i^0,\sigma),\;
\bigotimes_{i\in V}M(\delta_i^1,\sigma_1)\Bigr).
\]
The symmetric mixtures have matched second moments along the ridge,
$\E X^2=\delta_i^2+\sigma^2$ invariant, and all odd cumulants vanish; the leading mismatch
is in the fourth cumulant, $k_4(M(\delta,\sigma))=-2\delta^4$, whence
$\lvert\Delta k_4\rvert=2\lvert(\delta_i^1)^4-(\delta_i^0)^4\rvert
=2\sigma^2\varepsilon\,(2(\delta_i^0)^2+\sigma^2\varepsilon)
\le C\sigma^4\varepsilon(\lambda_i^0+\varepsilon)$,
and the higher even cumulants carry the same factor $\varepsilon(\lambda_i^0+\varepsilon)$.
Lemma~\ref{lem:mixture-chisq} below then gives the per-coordinate chi-square divergence
\[
\chi^2\bigl(M(\delta_i^1,\sigma_1)\,\big\|\,M(\delta_i^0,\sigma)\bigr)
\;\le\;C_2\,\varepsilon^2\bigl(\lambda_i^0+\varepsilon\bigr)^2 ,
\]
for a constant $C_2=C_2(C_\infty)$, valid while $\lambda_i^0\le4C_\infty^2$ and
$\varepsilon\le\tfrac12$; the lemma plays the role of the mixture chi-square
computations in the lower-bound literature for quadratic functionals
\citep{caiLow2006adaptive}. Tensorizing and using $\mathrm{KL}\le\log(1+\chi^2)\le\chi^2$,
\begin{align*}
\mathrm{KL}\bigl(P_{\delta^1,\sigma_1}\,\big\|\,P_{\delta^0,\sigma}\bigr)
&\;\le\;C_2\,\varepsilon^2\Bigl[\sum_{i\in V}(\lambda_i^0)^2
+2\varepsilon\sum_{i\in V}\lambda_i^0+m\varepsilon^2\Bigr]\\
&\;\le\;C_3\Bigl[\kappa^2\,\frac{\norm{\delta^0}_m^2}{\sigma_0^2}\,
+\kappa^3m^{-1/2}+\kappa^4m^{-1}\Bigr],
\end{align*}
where $\sum_i(\lambda_i^0)^2\le4C_\infty^2\sum_i\lambda_i^0$ and
$\sum_i\lambda_i^0\le 4m\norm{\delta^0}_m^2/\sigma_0^2$ were used. The first bracketed
term is the binding one: it vanishes only along shrinking profiles, which is why the
conclusion of Proposition~\ref{prop:floor} is restricted to $\norm{\delta^0}_m^2\le\rho_n$
with $\rho_n\to0$. Under that restriction the right-hand side is
$C_3[\kappa^2\rho_n/\sigma_0^2+\kappa^3m^{-1/2}+\kappa^4m^{-1}]\to0$ for fixed $\kappa$,
and by Pinsker $\mathrm{TV}(P_{\delta^0,\sigma},P_{\delta^1,\sigma_1})\le\tau_n(\kappa)\to0$.
(The restriction is the relevant regime: under Assumption~\ref{as:decay} the training-fold
profile satisfies $\norm{\delta^0}_m^2=O_P(n^{-2s/(2s+1)})$ by Lemma~\ref{lem:risk}, so it
holds with probability tending to one. For a profile fixed at the noise scale the same
argument still yields $\mathrm{TV}\le C\kappa<1-\alpha$ for $\kappa$ small, a weaker,
non-vanishing floor.)

\medskip\noindent\emph{Step 3: transfer.} Validity \eqref{eq:floor-class} at the paired
problem gives
$P_{\delta^1,\sigma_1}(\widehat U\ge\norm{\delta^1}_m^2)\ge1-\alpha$, and
$\norm{\delta^1}_m^2\ge\kappa\sigma_0^2m^{-1/2}$, so
\[
P_{\delta^0,\sigma}\bigl(\widehat U\ge\kappa\sigma_0^2m^{-1/2}\bigr)
\;\ge\;
P_{\delta^1,\sigma_1}\bigl(\widehat U\ge\kappa\sigma_0^2m^{-1/2}\bigr)-\tau_n(\kappa)
\;\ge\;1-\alpha-\tau_n(\kappa),
\]
which is \eqref{eq:floor}. Dividing by
$\mathcal{E}(\widehat f)=O_P(n^{-2s/(2s+1)})$ and using $m\asymp n$ gives the rate
\eqref{eq:cvdiverge}. \qed

\medskip\noindent\emph{Remarks.} (i) The proposition is stated conditionally on the
training fold; repeated splits are handled by conditioning on all split assignments,
which are independent of the data. (ii) The bound $\widehat U$ may depend arbitrarily on
the training responses and the split structure, since these are fixed by the
conditioning; only the access to the held-out noise through the losses is restricted.
(iii) The per-coordinate chi-square bound in Step 2 is the technical heart;
Lemma~\ref{lem:mixture-chisq} proves it by interpolation along the ridge.

\begin{lemma}[Mixture chi-square along the ridge]
\label{lem:mixture-chisq}
Let $\Lambda\ge1$, $0\le\lambda\le\Lambda$, and $\varepsilon\in(0,\tfrac12]$. For
$t\in[0,\varepsilon]$ let $p_t$ be the density of the symmetric Gaussian mixture
$\tfrac12\mathcal N(b_t,s_t^2)+\tfrac12\mathcal N(-b_t,s_t^2)$ with
$b_t:=\sqrt{\lambda+t}$ and $s_t^2:=1-t$. There is a constant $C(\Lambda)$ with
\[
\chi^2\bigl(p_\varepsilon\,\big\|\,p_0\bigr)\;\le\;C(\Lambda)\,\varepsilon^2(\lambda+\varepsilon)^2 .
\]
Since chi-square divergence is invariant under the common rescaling $x\mapsto x/\sigma$,
which maps the pair $\bigl(M(\delta_i^1,\sigma_1),\,M(\delta_i^0,\sigma)\bigr)$ of Step 2
onto $(p_\varepsilon,p_0)$ with $\lambda=\lambda_i^0$, the display of Step 2 holds with
$C_2:=C(4C_\infty^2)$.
\end{lemma}

\begin{proof}
Write $h=h_t:=1/s_t^2=1/(1-t)\in[1,2]$, so that
\[
p_t(x)=(h/2\pi)^{1/2}\exp\{-h(x^2+b_t^2)/2\}\cosh(hb_tx),
\]
jointly smooth in $(t,x)$ with
Gaussian decay. Differentiating with $\tfrac{d}{dt}\log h=h$, $\tfrac{d}{dt}b_t^2=1$,
and $\tfrac{d}{dt}(hb_t)=h^2b_t+h/(2b_t)$ gives $\partial_tp_t=S_tp_t$ with score
\[
S_t(x)\;=\;-\frac{h^2}{2}\bigl(x^2+b_t^2\bigr)
+\Bigl(h^2b_t+\frac{h}{2b_t}\Bigr)\,x\tanh(hb_tx) .
\]

\emph{Step A (interpolation).} Since $p_\varepsilon-p_0=\int_0^\varepsilon S_tp_t\,dt$,
Minkowski's integral inequality in $L^2(1/p_0)$ gives
\[
\chi^2(p_\varepsilon\|p_0)^{1/2}
=\Bigl(\int\frac{(p_\varepsilon-p_0)^2}{p_0}\Bigr)^{1/2}
\le\int_0^\varepsilon\Bigl(\int S_t^2\,\frac{p_t^2}{p_0}\Bigr)^{1/2}dt
\le\int_0^\varepsilon\Bigl(\sup_x\frac{p_t}{p_0}\Bigr)^{1/2}
\bigl(\E_{p_t}S_t^2\bigr)^{1/2}dt .
\]

\emph{Step B (ratio bound): $\sup_xp_t/p_0\le\sqrt2\,\exp(2\Lambda+2)$.} From the closed
form,
\[
\log\frac{p_t(x)}{p_0(x)}
=\tfrac12\log h-\frac{(h-1)x^2+(hb_t^2-\lambda)}{2}
+\log\cosh(hb_tx)-\log\cosh(\sqrt\lambda\,x) .
\]
Since $hb_t\ge\sqrt\lambda$ and $\lvert\tfrac{d}{du}\log\cosh u\rvert\le1$, the last
difference is at most $(hb_t-\sqrt\lambda)|x|$, and $hb_t^2-\lambda\ge0$, so maximizing
the quadratic in $|x|$,
\[
\log\frac{p_t}{p_0}\le\tfrac12\log2-\frac{(h-1)x^2}{2}+(hb_t-\sqrt\lambda)|x|
\le\tfrac12\log2+\frac{(hb_t-\sqrt\lambda)^2}{2(h-1)} .
\]
Now $h-1=th\ge t$, and
$hb_t-\sqrt\lambda=(h-1)b_t+(b_t-\sqrt\lambda)
\le2t\sqrt{\lambda+t}+t/\sqrt{\lambda+t}\le2t\sqrt{\lambda+t}+\sqrt t$,
so $(hb_t-\sqrt\lambda)^2\le8t^2(\lambda+t)+2t$ and
$(hb_t-\sqrt\lambda)^2/(2(h-1))\le4t(\lambda+t)+1\le2\Lambda+2$ for $t\le\tfrac12$,
$\lambda\le\Lambda$.

\emph{Step C (score moment): $\E_{p_t}S_t^2\le C_0(\lambda+t)^2$ for an absolute
constant $C_0$.} If $b_t^2=\lambda+t\ge1$, then crudely
$|S_t|\le2(x^2+b_t^2)+(4b_t+1)|x|$ using $h\le2$, $\tanh\le1$, and
$1/(2b_t)\le\tfrac12$; under $p_t$, $\E X^2=b_t^2+s_t^2\le2b_t^2$ and
$\E X^4\le10b_t^4$, so $\E S_t^2\le Cb_t^4=C(\lambda+t)^2$. If $b_t^2\le1$, the
third-order expansion $\tanh u=u-u^3/3+r(u)$ with $|r(u)|\le\tfrac2{15}|u|^5$,
substituted at $u=hb_tx$, cancels the $(h^2/2)x^2$ term of $S_t$ exactly and leaves
\[
S_t=h^2b_t^2\Bigl(hx^2-\tfrac12\Bigr)
-\Bigl(\frac{h^5b_t^4}{3}+\frac{h^4b_t^2}{6}\Bigr)x^4+\rho(x),
\qquad|\rho(x)|\le Cb_t^4x^6 ,
\]
the bound on $\rho$ using $b_t\le1$ and $h\le2$. Under $p_t$ with $b_t\le1$ and
$s_t^2\in[\tfrac12,1]$, all moments of $X$ up to order twelve are bounded by absolute
constants, so
\[
\E_{p_t}S_t^2
\le3\Bigl[h^4b_t^4\,\E\bigl(hX^2-\tfrac12\bigr)^2
+Cb_t^4\,\E X^8+C b_t^8\,\E X^{12}\Bigr]
\le C_0\,b_t^4=C_0(\lambda+t)^2 .
\]

\emph{Step D (combining).} By Steps A--C and
$\int_0^\varepsilon(\lambda+t)\,dt=\lambda\varepsilon+\varepsilon^2/2\le\varepsilon(\lambda+\varepsilon)$,
\[
\chi^2(p_\varepsilon\|p_0)
\le\sqrt2\,\exp(2\Lambda+2)\,C_0\Bigl[\int_0^\varepsilon(\lambda+t)\,dt\Bigr]^2
\le\sqrt2\,\exp(2\Lambda+2)\,C_0\,\varepsilon^2(\lambda+\varepsilon)^2 . \qedhere
\]
\end{proof}

\subsection{Proof of Corollary~\ref{cor:iid} (i.i.d.\ noise)}
\label{app:proof-iid}

Let $w_1,\dots,w_n$ be independent, symmetric, sub-Gaussian, independent of the design,
with $\Var(w_i)\ge\tau_0^2$ and $\Var(w_i^2)\ge\kappa_0$. Symmetry gives conditional
symmetry of the signs. For Assumption~\ref{as:energy}: in the homoscedastic case
$\sigma_i\equiv\sigma$, $\E_\varepsilon\norm{Hw}_n^2=\sigma^2\operatorname{tr}(H^2)/n
=\sigma^2 c(s)d_n/n$ exactly with $c(s):=(2s-1)/(2s)$; in the heteroscedastic case with
$\sigma_i\ge\tau_0$,
$\tfrac1n\sum_i\sigma_i^2(H^\top H)_{ii}\ge\tau_0^2\operatorname{tr}(H^2)/n=\tau_0^2 c(s)d_n/n$
regardless of any coupling between $\sigma_i$ and the leverage, so
Assumption~\ref{as:energy} holds with $c_0=\tau_0^2 c(s)$. For
Assumption~\ref{as:spread}, the law of large numbers gives
$\operatorname{empvar}_i\{w_i^2\}\to\Var(w_1^2)\ge\kappa_0$ almost surely. Both assumptions
holding with probability tending to one, Theorem~\ref{thm:rate}(ii) and
Proposition~\ref{prop:cv} apply, giving the two unconditional conclusions. Under two-point
noise $w_i^2\equiv\sigma^2$ the variance $\Var(w_i^2)=0$, Assumption~\ref{as:spread} fails,
and cross-validation is rate-efficient, the single disclosed exception. \qed

\section{Additional experimental details}
\label{app:sim-details}

This section records the data-generating process, the calibration settings, and the
three-regime boundary summarized in \S\ref{sec:simulation}.

\subsection{Data-generating process}

The design is fixed, $x_i=(i-1)/(n-1)$ on $[0,1]$, with regression function
$f^*(x)=\sin(2\pi x)+\tfrac12\sin(6\pi x)$ and radial basis kernel
$k(x,x')=\exp(-\gamma(x-x')^2)$, $\gamma=50$, ridge penalty $\lambda=10^{-3}$. Four noise
laws are used, all conditionally symmetric: heteroscedastic Gaussian with
$\sigma(x)=0.1+0.9x$; Laplace of unit scale; Student-$t_4$; and standard Cauchy, whose
variance is infinite. Each cell is repeated over $R$ replicates, $R=160$ for $n\le2000$,
$R=120$ at $n=4000$, and $R=100$ at $n=8000$; the true $95\%$ prediction-error quantile
$q_{0.95}$ is the empirical quantile of $\{\mathcal{E}^{(r)}\}_{r=1}^R$.

\subsection{Calibration}

Each bound is calibrated with $L=199$ Gaussian draws and the order statistic
\eqref{eq:orderstat} at $k^*=\lceil(L+1)\cdot0.95\rceil=190$. The data-driven envelope is
$a_i=|\widetilde w_i|/(1-h_{ii})+|\widehat b_i|$, the leverage correction of the wild
bootstrap with $h_{ii}$ the $i$th diagonal of $H$, and $\widehat b$ the bias vector of an
undersmoothed pilot fit at penalty $\lambda/20$. Cross-validation is the five-fold
hold-out bound \eqref{eq:cvbound} with the $t_4$ critical value; SURE and the Rademacher
wild-refit statistic \citep{wainwright2025wild} are computed on the same replicate stream.
The strongest cross-validation baseline of \S\ref{sec:sim-main} is the nested
cross-validation of \citet{batesHastieTibshirani2024cv} on the small-sample cells, drawn
on the same stream; the portability run of \S\ref{sec:sim-main} replaces the linear
smoother by the reproducing-kernel-ball-constrained least-squares fit, re-solved at each
refit draw and cross-validation fold. The rigid-fit misspecification of
\S\ref{app:sim-scope} uses the same triangle-wave target with penalty $10^{-2}$.

\subsection{Metrics}

We report the median-based summaries \eqref{eq:metrics}. Median and interquartile range
replace the mean and coefficient of variation because the cross-validation bound has no
finite variance under the heavy-tailed laws, where a single replicate can move a mean-based
summary by tens of percent; the median and interquartile range are stable under the same
draws. The reported accuracy is $\operatorname{med}(\widehat U)/q_{0.95}$ and the reported
coverage is the fraction of replicates with $\widehat U^{(r)}\ge\mathcal{E}^{(r)}$.

Table~\ref{tab:fullgrid} reports the full grid underlying Figure~\ref{fig:sim}, in the
six-method scheme of the main text's Table~2. The Gaussian refit holds accuracy near
$1.8\times$ with coverage at or above the stated level across every sample size and noise
law. The cross-validation margin grows with $n$ in all three forms; with the noise level
known its coverage falls below $0.95$, most sharply under Student-$t_4$, while the
estimated level and the quantile repair cover fully at the same or greater width. Under
standard Cauchy noise ($n=2000$), whose $95\%$ quantile is dominated by extreme
realizations, the ratio to each dataset's own realized error is $4.6\times$ for the
Gaussian refit, over $500\times$ for cross-validation with the noise level known,
$286\times$ with it estimated, and $547\times$ for the quantile repair, all at full
coverage.

\begin{table}[t]
\centering
\caption{Full simulation grid: accuracy (median bound divided by the true $95\%$
prediction-error quantile) with empirical coverage in parentheses, radial basis kernel,
kernel ridge. Cross-validation appears in the three forms of the main text's Table~2;
the tightest method with coverage at least $0.95$ is in bold.}
\label{tab:fullgrid}
\footnotesize
\setlength{\tabcolsep}{3pt}
\begin{tabular}{@{}llcccccc@{}}
\toprule
noise & $n$ & Gaussian refit & CV ($\bar\sigma^2$) & CV ($\widehat\sigma^2$) & quantile CV & SURE & Rademacher \\
\midrule
Gaussian & $500$ & $\bm{1.73\ (1.00)}$ & $4.33\ (0.92)$ & $4.56\ (1.00)$ & $7.60\ (1.00)$ & $0.37\ (0.49)$ & $0.23\ (0.36)$ \\
Gaussian & $1000$ & $\bm{1.73\ (0.99)}$ & $5.58\ (0.94)$ & $5.55\ (1.00)$ & $9.76\ (1.00)$ & $0.46\ (0.49)$ & $0.31\ (0.44)$ \\
Gaussian & $2000$ & $\bm{1.83\ (1.00)}$ & $7.61\ (0.93)$ & $7.07\ (1.00)$ & $13.5\ (1.00)$ & $0.95\ (0.58)$ & $0.31\ (0.36)$ \\
Gaussian & $4000$ & $\bm{1.91\ (0.99)}$ & $10.0\ (0.97)$ & $10.1\ (1.00)$ & $17.2\ (1.00)$ & $0.55\ (0.50)$ & $0.43\ (0.48)$ \\
Gaussian & $8000$ & $\bm{1.83\ (1.00)}$ & $10.5\ (0.96)$ & $10.1\ (1.00)$ & $19.1\ (1.00)$ & $-0.14\ (0.44)$ & $0.42\ (0.39)$ \\
\addlinespace
Laplace & $500$ & $\bm{1.88\ (1.00)}$ & $6.08\ (0.91)$ & $6.25\ (1.00)$ & $11.1\ (1.00)$ & $0.36\ (0.45)$ & $0.33\ (0.33)$ \\
Laplace & $1000$ & $\bm{1.80\ (1.00)}$ & $8.16\ (0.94)$ & $8.09\ (1.00)$ & $14.2\ (1.00)$ & $0.45\ (0.49)$ & $0.31\ (0.34)$ \\
Laplace & $2000$ & $\bm{1.89\ (1.00)}$ & $11.3\ (0.93)$ & $11.1\ (1.00)$ & $20.2\ (1.00)$ & $0.30\ (0.46)$ & $0.30\ (0.31)$ \\
Laplace & $4000$ & $\bm{1.74\ (0.99)}$ & $15.4\ (0.92)$ & $13.8\ (1.00)$ & $23.5\ (1.00)$ & $1.23\ (0.53)$ & $0.28\ (0.28)$ \\
Laplace & $8000$ & $\bm{1.92\ (1.00)}$ & $18.2\ (0.98)$ & $18.9\ (1.00)$ & $31.1\ (1.00)$ & $0.36\ (0.49)$ & $0.36\ (0.32)$ \\
\addlinespace
Student-$t_4$ & $500$ & $\bm{1.88\ (1.00)}$ & $6.67\ (0.88)$ & $7.61\ (1.00)$ & $13.3\ (1.00)$ & $0.11\ (0.43)$ & $0.27\ (0.32)$ \\
Student-$t_4$ & $1000$ & $\bm{1.70\ (1.00)}$ & $9.16\ (0.89)$ & $10.5\ (1.00)$ & $17.6\ (1.00)$ & $-0.85\ (0.37)$ & $0.32\ (0.36)$ \\
Student-$t_4$ & $2000$ & $\bm{1.76\ (1.00)}$ & $13.6\ (0.90)$ & $14.8\ (1.00)$ & $26.6\ (1.00)$ & $-0.40\ (0.43)$ & $0.28\ (0.30)$ \\
Student-$t_4$ & $4000$ & $\bm{1.79\ (1.00)}$ & $19.3\ (0.88)$ & $20.4\ (1.00)$ & $38.0\ (1.00)$ & $-2.35\ (0.44)$ & $0.23\ (0.28)$ \\
Student-$t_4$ & $8000$ & $\bm{1.84\ (1.00)}$ & $21.3\ (0.85)$ & $28.4\ (1.00)$ & $50.7\ (1.00)$ & $-3.93\ (0.39)$ & $0.27\ (0.29)$ \\
\bottomrule
\end{tabular}
\end{table}

\subsection{The three regimes}
\label{app:sim-scope}

The scope of \S\ref{sec:theory} is delimited by the spectral decay and by whether $f^*$
lies in the reproducing kernel Hilbert space. Table~\ref{tab:regimefull} gives the
heteroscedastic-Gaussian grid across sample sizes; the pattern is the same under the other
noise laws.

\paragraph{Fast decay, well specified.} For the radial basis design with $f^*\in\HH$ the
Gaussian refit holds full coverage at $1.7$--$1.9\times$ the true quantile while the
cross-validation ratio grows with $n$; this is the regime of Table~\ref{tab:sim} and
Figure~\ref{fig:sim}.

\paragraph{Slow decay.} As the eigenvalues decay more slowly the bound stays valid but
loses its margin. For the Mat\'ern-$3/2$ kernel ($s=2$) it holds $2.7$--$3.3\times$ at full
coverage, still inside cross-validation's range; for the Mat\'ern-$1/2$ kernel ($s=1$) it is
valid but no longer tighter than cross-validation at moderate $n$. Proposition~\ref{prop:cv}
accounts for the transition: the cross-validation ratio diverges at exponent
$(2s-1)/(2(2s+1))$, which is $1/6$ at $s=1$, so slowly that it overtakes the bound only at
the largest $n$.

\paragraph{Misspecified.} When $f^*\notin\HH$ the outcome depends on how rigidly the target
is fit. The data-driven envelope is built from the residuals, which carry the pointwise
bias, so a flexible fit (penalty $10^{-3}$) keeps the bound covered, eroding only to $0.83$
at $n=8000$. An over-smoothed fit (penalty $10^{-2}$) cannot track the target: the error
becomes bias-dominated, the envelope no longer captures it, and coverage falls to $0.10$.
In this bias-dominated regime cross-validation, which estimates the realized loss directly,
is the appropriate tool.

\begin{table}[t]
\centering
\caption{Regime grid, heteroscedastic Gaussian noise: accuracy with coverage in
parentheses, the full sample-size sweep of the regime study. Cross-validation appears
in the three forms of the main text's Table~2; per column the tightest method with
coverage at least $0.95$ is in bold, and the misspecified rows have no comparator.
As eigendecay slows the Gaussian refit keeps full coverage but loses its margin:
on the Mat\'ern-$1/2$ kernel the practical cross-validation forms are tighter at
moderate $n$, and the bound overtakes only at the largest $n$. Under a rigid
over-smoothed fit on a misspecified target its coverage falls with $n$.}
\label{tab:regimefull}
\footnotesize
\setlength{\tabcolsep}{3pt}
\begin{tabular}{@{}llccccc@{}}
\toprule
regime & & $n{=}500$ & $1000$ & $2000$ & $4000$ & $8000$ \\
\midrule
Mat\'ern-$3/2$ & ours & $\bm{3.29\ (1.00)}$ & $\bm{3.01\ (1.00)}$ & $\bm{2.72\ (1.00)}$ & $\bm{2.85\ (1.00)}$ & $\bm{2.69\ (1.00)}$ \\
 & CV ($\bar\sigma^2$) & $3.64\ (0.91)$ & $4.29\ (0.95)$ & $6.01\ (0.94)$ & $8.58\ (0.97)$ & $10.6\ (0.98)$ \\
 & CV ($\widehat\sigma^2$) & $3.81\ (1.00)$ & $4.45\ (1.00)$ & $5.50\ (1.00)$ & $8.54\ (1.00)$ & $9.94\ (1.00)$ \\
 & quantile CV & $6.26\ (1.00)$ & $7.67\ (1.00)$ & $10.1\ (1.00)$ & $14.5\ (1.00)$ & $18.7\ (1.00)$ \\
\addlinespace
Mat\'ern-$1/2$ & ours & $8.57\ (1.00)$ & $7.02\ (1.00)$ & $6.08\ (1.00)$ & $6.04\ (1.00)$ & $\bm{5.22\ (1.00)}$ \\
 & CV ($\bar\sigma^2$) & $2.52\ (0.91)$ & $\bm{3.06\ (0.96)}$ & $4.07\ (0.94)$ & $\bm{5.79\ (0.97)}$ & $7.01\ (0.97)$ \\
 & CV ($\widehat\sigma^2$) & $\bm{2.74\ (1.00)}$ & $3.23\ (1.00)$ & $\bm{3.84\ (1.00)}$ & $5.92\ (1.00)$ & $6.44\ (1.00)$ \\
 & quantile CV & $4.25\ (1.00)$ & $5.12\ (1.00)$ & $6.69\ (1.00)$ & $9.82\ (1.00)$ & $11.9\ (1.00)$ \\
\addlinespace
misspec, flexible & ours & $1.76\ (1.00)$ & $1.57\ (0.99)$ & $1.43\ (0.99)$ & $1.21\ (0.97)$ & $1.05\ (0.83)$ \\
\addlinespace
misspec, rigid & ours & $1.45\ (1.00)$ & $1.25\ (0.98)$ & $1.06\ (0.91)$ & $0.89\ (0.56)$ & $0.81\ (0.10)$ \\
\bottomrule
\end{tabular}
\end{table}

\bibliography{references}

\begin{thebibliography}{45}
\providecommand{\natexlab}[1]{#1}
\providecommand{\url}[1]{\texttt{#1}}
\expandafter\ifx\csname urlstyle\endcsname\relax
  \providecommand{\doi}[1]{doi: #1}\else
  \providecommand{\doi}{doi: \begingroup \urlstyle{rm}\Url}\fi

\bibitem[Abbasi-Yadkori et~al.(2011)Abbasi-Yadkori, P{\'a}l, and
  Szepesv{\'a}ri]{abbasiYadkori2011improved}
Y.~Abbasi-Yadkori, D.~P{\'a}l, and C.~Szepesv{\'a}ri.
\newblock Improved algorithms for linear stochastic bandits.
\newblock In \emph{Adv. Neural Inf. Process. Syst.}, volume~24, 2011.

\bibitem[Anderson(1955)]{anderson1955integral}
T.~W. Anderson.
\newblock The integral of a symmetric unimodal function over a symmetric convex
  set and some probability inequalities.
\newblock \emph{Proc. Amer. Math. Soc.}, 6\penalty0 (2):\penalty0 170--176,
  1955.

\bibitem[Austern and Zhou(2025)]{austernZhou2025cv}
M.~Austern and W.~Zhou.
\newblock Asymptotics of cross-validation.
\newblock \emph{Ann. Inst. Henri Poincar{\'e} Probab. Stat.}, 61\penalty0 (4),
  2025.

\bibitem[Baraud(2004)]{baraud2004confidence}
Y.~Baraud.
\newblock Confidence balls in {G}aussian regression.
\newblock \emph{Ann. Statist.}, 32\penalty0 (2):\penalty0 528--551, 2004.

\bibitem[Barber et~al.(2021)Barber, Cand{\`e}s, Ramdas, and
  Tibshirani]{barber2021predictive}
R.~F. Barber, E.~J. Cand{\`e}s, A.~Ramdas, and R.~J. Tibshirani.
\newblock Predictive inference with the jackknife+.
\newblock \emph{Ann. Statist.}, 49\penalty0 (1):\penalty0 486--507, 2021.

\bibitem[Bates et~al.(2024)Bates, Hastie, and
  Tibshirani]{batesHastieTibshirani2024cv}
S.~Bates, T.~Hastie, and R.~Tibshirani.
\newblock Cross-validation: what does it estimate and how well does it do it?
\newblock \emph{J. Amer. Statist. Assoc.}, 119\penalty0 (546):\penalty0
  1434--1445, 2024.

\bibitem[Bayle et~al.(2020)Bayle, Bayle, Janson, and Mackey]{bayle2020cv}
P.~Bayle, A.~Bayle, L.~Janson, and L.~Mackey.
\newblock Cross-validation confidence intervals for test error.
\newblock In \emph{Adv. Neural Inf. Process. Syst.}, volume~33, 2020.

\bibitem[Bellec and Zhang(2021)]{bellecZhang2021sure}
P.~C. Bellec and C.-H. Zhang.
\newblock Second-order {S}tein: {SURE} for {SURE} and other applications in
  high-dimensional inference.
\newblock \emph{Ann. Statist.}, 49\penalty0 (4):\penalty0 1864--1903, 2021.

\bibitem[Bengio and Grandvalet(2004)]{bengioGrandvalet2004no}
Y.~Bengio and Y.~Grandvalet.
\newblock No unbiased estimator of the variance of {K}-fold cross-validation.
\newblock \emph{J. Mach. Learn. Res.}, 5:\penalty0 1089--1105, 2004.

\bibitem[Beran and D{\"u}mbgen(1998)]{beranDumbgen1998modulation}
R.~Beran and L.~D{\"u}mbgen.
\newblock Modulation of estimators and confidence sets.
\newblock \emph{Ann. Statist.}, 26\penalty0 (5):\penalty0 1826--1856, 1998.

\bibitem[Cai and Low(2006)]{caiLow2006adaptive}
T.~T. Cai and M.~G. Low.
\newblock Adaptive confidence balls.
\newblock \emph{Ann. Statist.}, 34\penalty0 (1):\penalty0 202--228, 2006.

\bibitem[Caponnetto and De~Vito(2007)]{caponnettoDeVito2007optimal}
A.~Caponnetto and E.~De~Vito.
\newblock Optimal rates for the regularized least-squares algorithm.
\newblock \emph{Found. Comput. Math.}, 7\penalty0 (3):\penalty0 331--368, 2007.

\bibitem[Chowdhury and Gopalan(2017)]{chowdhuryGopalan2017kernelized}
S.~R. Chowdhury and A.~Gopalan.
\newblock On kernelized multi-armed bandits.
\newblock In \emph{Proc. 34th Int. Conf. Machine Learning}, volume~70 of
  \emph{Proceedings of Machine Learning Research}, pages 844--853, 2017.

\bibitem[Cs{\'a}ji and Kis(2019)]{csajiKis2019kernel}
B.~C. Cs{\'a}ji and K.~B. Kis.
\newblock Distribution-free uncertainty quantification for kernel methods by
  gradient perturbations.
\newblock \emph{Machine Learning}, 108:\penalty0 1677--1699, 2019.

\bibitem[Cs{\'a}ji et~al.(2015)Cs{\'a}ji, Campi, and
  Weyer]{csajiCampiWeyer2015sps}
B.~C. Cs{\'a}ji, M.~C. Campi, and E.~Weyer.
\newblock Sign-perturbed sums: a new system identification approach for
  constructing exact non-asymptotic confidence regions in linear regression
  models.
\newblock \emph{IEEE Trans. Signal Process.}, 63\penalty0 (1):\penalty0
  169--181, 2015.

\bibitem[Davidson and Flachaire(2008)]{davidson2008wild}
R.~Davidson and E.~Flachaire.
\newblock The wild bootstrap, tamed at last.
\newblock \emph{Journal of Econometrics}, 146\penalty0 (1):\penalty0 162--169,
  2008.

\bibitem[D{\"o}bler and Peccati(2017)]{doeblerPeccati2017}
C.~D{\"o}bler and G.~Peccati.
\newblock The fourth moment theorem on the {P}oisson space.
\newblock \emph{Ann. Probab.}, 45\penalty0 (3):\penalty0 1804--1849, 2017.

\bibitem[Fiedler et~al.(2021)Fiedler, Scherer, and
  Trimpe]{fiedler2021practical}
C.~Fiedler, C.~W. Scherer, and S.~Trimpe.
\newblock Practical and rigorous uncertainty bounds for {G}aussian process
  regression.
\newblock In \emph{Proc. AAAI Conf. Artificial Intelligence}, volume~35, pages
  7439--7447, 2021.

\bibitem[Geisser(1975)]{geisser1975predictive}
S.~Geisser.
\newblock The predictive sample reuse method with applications.
\newblock \emph{J. Amer. Statist. Assoc.}, 70\penalty0 (350):\penalty0
  320--328, 1975.

\bibitem[Hu and Simchi-Levi(2025{\natexlab{a}})]{huSimchiLevi2025bregman}
H.~Hu and D.~Simchi-Levi.
\newblock Perturbing the derivative: Wild refitting for model-free evaluation
  of machine learning models under {B}regman losses.
\newblock arXiv preprint arXiv:2509.02476, 2025{\natexlab{a}}.

\bibitem[Hu and Simchi-Levi(2025{\natexlab{b}})]{huSimchiLevi2025doubly}
H.~Hu and D.~Simchi-Levi.
\newblock Perturbing the derivative: Doubly wild refitting for model-free
  evaluation of opaque machine learning predictors.
\newblock arXiv preprint arXiv:2511.18789, 2025{\natexlab{b}}.

\bibitem[Imhof(1961)]{imhof1961computing}
J.~P. Imhof.
\newblock Computing the distribution of quadratic forms in normal variables.
\newblock \emph{Biometrika}, 48\penalty0 (3/4):\penalty0 419--426, 1961.

\bibitem[Juditsky and Lambert-Lacroix(2003)]{juditskyLambertLacroix2003}
A.~Juditsky and S.~Lambert-Lacroix.
\newblock Nonparametric confidence set estimation.
\newblock \emph{Math. Methods Statist.}, 12\penalty0 (4):\penalty0 410--428,
  2003.

\bibitem[Lahr et~al.(2025)Lahr, K{\"o}hler, Scampicchio, and
  Zeilinger]{lahr2025optimal}
A.~Lahr, J.~K{\"o}hler, A.~Scampicchio, and M.~N. Zeilinger.
\newblock Optimal kernel regression bounds under energy-bounded noise.
\newblock arXiv:2505.22235, 2025.

\bibitem[Laurent and Massart(2000)]{laurentMassart2000}
B.~Laurent and P.~Massart.
\newblock Adaptive estimation of a quadratic functional by model selection.
\newblock \emph{Ann. Statist.}, 28\penalty0 (5):\penalty0 1302--1338, 2000.

\bibitem[Lei(2020)]{lei2020cvc}
J.~Lei.
\newblock Cross-validation with confidence.
\newblock \emph{J. Amer. Statist. Assoc.}, 115\penalty0 (532):\penalty0
  1978--1997, 2020.

\bibitem[Lei et~al.(2018)Lei, G'Sell, Rinaldo, Tibshirani, and
  Wasserman]{lei2018distribution}
J.~Lei, M.~G'Sell, A.~Rinaldo, R.~J. Tibshirani, and L.~Wasserman.
\newblock Distribution-free predictive inference for regression.
\newblock \emph{J. Amer. Statist. Assoc.}, 113\penalty0 (523):\penalty0
  1094--1111, 2018.

\bibitem[Li(1989)]{li1989honest}
K.-C. Li.
\newblock Honest confidence regions for nonparametric regression.
\newblock \emph{Ann. Statist.}, 17\penalty0 (3):\penalty0 1001--1008, 1989.

\bibitem[Lin et~al.(2020)Lin, Rudi, Rosasco, and
  Cevher]{linCevherRosasco2018optimal}
J.~Lin, A.~Rudi, L.~Rosasco, and V.~Cevher.
\newblock Optimal rates for spectral algorithms with least-squares regression
  over {H}ilbert spaces.
\newblock \emph{Appl. Comput. Harmon. Anal.}, 48\penalty0 (3):\penalty0
  868--890, 2020.

\bibitem[Liu(1988)]{liu1988bootstrap}
R.~Y. Liu.
\newblock Bootstrap procedures under some non-i.i.d. models.
\newblock \emph{Ann. Statist.}, 16\penalty0 (4):\penalty0 1696--1708, 1988.

\bibitem[Maddalena et~al.(2021)Maddalena, Scharnhorst, and
  Jones]{maddalena2021deterministic}
E.~T. Maddalena, P.~Scharnhorst, and C.~N. Jones.
\newblock Deterministic error bounds for kernel-based learning techniques under
  bounded noise.
\newblock \emph{Automatica}, 134:\penalty0 109896, 2021.

\bibitem[Mammen(1992)]{mammen1992when}
E.~Mammen.
\newblock \emph{When Does Bootstrap Work? Asymptotic Results and Simulations},
  volume~77 of \emph{Lecture Notes in Statistics}.
\newblock Springer-Verlag, New York, 1992.

\bibitem[Mammen(1993)]{mammen1993bootstrap}
E.~Mammen.
\newblock Bootstrap and wild bootstrap for high dimensional linear models.
\newblock \emph{Ann. Statist.}, 21\penalty0 (1):\penalty0 255--285, 1993.

\bibitem[Nourdin et~al.(2010)Nourdin, Peccati, and
  Reinert]{nourdinPeccatiReinert2010}
I.~Nourdin, G.~Peccati, and G.~Reinert.
\newblock Invariance principles for homogeneous sums: universality of
  {G}aussian {W}iener chaos.
\newblock \emph{Ann. Probab.}, 38\penalty0 (5):\penalty0 1947--1985, 2010.

\bibitem[Robins and van~der Vaart(2006)]{robinsVanDerVaart2006adaptive}
J.~Robins and A.~van~der Vaart.
\newblock Adaptive nonparametric confidence sets.
\newblock \emph{Ann. Statist.}, 34\penalty0 (1):\penalty0 229--253, 2006.

\bibitem[Rudelson and Vershynin(2013)]{rudelsonVershynin2013hansonwright}
M.~Rudelson and R.~Vershynin.
\newblock Hanson--{W}right inequality and sub-{G}aussian concentration.
\newblock \emph{Electron. Commun. Probab.}, 18:\penalty0 1--9, 2013.

\bibitem[Singh and Vijaykumar(2023)]{singhVijaykumar2023kernel}
R.~Singh and S.~Vijaykumar.
\newblock Kernel ridge regression inference.
\newblock arXiv:2302.06578, 2023.

\bibitem[Stein(1981)]{stein1981estimation}
C.~M. Stein.
\newblock Estimation of the mean of a multivariate normal distribution.
\newblock \emph{Ann. Statist.}, 9\penalty0 (6):\penalty0 1135--1151, 1981.

\bibitem[Stone(1982)]{stone1982optimal}
C.~J. Stone.
\newblock Optimal global rates of convergence for nonparametric regression.
\newblock \emph{Ann. Statist.}, 10\penalty0 (4):\penalty0 1040--1053, 1982.

\bibitem[Stone(1974)]{stone1974cross}
M.~Stone.
\newblock Cross-validatory choice and assessment of statistical predictions.
\newblock \emph{J. R. Statist. Soc. B}, 36\penalty0 (2):\penalty0 111--147,
  1974.

\bibitem[Vovk et~al.(2005)Vovk, Gammerman, and Shafer]{vovk2005algorithmic}
V.~Vovk, A.~Gammerman, and G.~Shafer.
\newblock \emph{Algorithmic Learning in a Random World}.
\newblock Springer, 2005.

\bibitem[Wager(2020)]{wager2020cv}
S.~Wager.
\newblock Cross-validation, risk estimation, and model selection: comment on a
  paper by {R}osset and {T}ibshirani.
\newblock \emph{J. Amer. Statist. Assoc.}, 115\penalty0 (529):\penalty0
  157--160, 2020.

\bibitem[Wainwright(2025)]{wainwright2025wild}
M.~J. Wainwright.
\newblock Wild refitting for black box prediction.
\newblock arXiv preprint arXiv:2506.21460, 2025.

\bibitem[Wu(1986)]{wu1986jackknife}
C.~F.~J. Wu.
\newblock Jackknife, bootstrap and other resampling methods in regression
  analysis.
\newblock \emph{Ann. Statist.}, 14\penalty0 (4):\penalty0 1261--1295, 1986.

\bibitem[Yang and Barron(1999)]{yang1999information}
Y.~Yang and A.~Barron.
\newblock Information-theoretic determination of minimax rates of convergence.
\newblock \emph{Ann. Statist.}, 27\penalty0 (5):\penalty0 1564--1599, 1999.

\end{thebibliography}

\end{document}